\documentclass[%
 reprint,
 amsmath,amssymb,
 aps,
]{revtex4-2}

\usepackage{graphicx}
\usepackage{dcolumn}
\usepackage{bm}
\usepackage{amsmath}
\usepackage{amsfonts}
\usepackage{amsthm}
\newtheorem{theorem}{Theorem}
\newtheorem{corollary}[theorem]{Corollary}
\newtheorem{lemma}[theorem]{Lemma}
\newtheorem{remark}[theorem]{Remark}
\newtheorem{definition}[theorem]{Definition}

\newtheorem{conjecture}[theorem]{Conjecture}
\usepackage{xcolor}
\usepackage{mleftright}
\definecolor{darkblue}{RGB}{0, 0, 255}
\definecolor{reddish}{rgb}{.8, 0.2, 0.2}
\usepackage[hidelinks]{hyperref}
\hypersetup{
	colorlinks=true,
	citecolor=darkblue,
	linkcolor=reddish,
	urlcolor=darkblue,
	pdfauthor={},
	pdfsubject={}
}
\usepackage{physics}
\usepackage{mathtools}
\usepackage{float}

\usepackage{xcolor}

\usepackage{aliascnt}
\usepackage{hyperref}
\usepackage{cleveref}
\usepackage{float}

\usepackage{comment}
\DeclareMathOperator{\Sym}{Sym}
\DeclareMathOperator{\GL}{GL}
\DeclareMathOperator{\diag}{diag}

\NewDocumentCommand{\pDist}{ O{M} O{\ensuremath{\eta} } m }{\ensuremath{p_{#3}  }}
\NewDocumentCommand{\piDist}{ O{N} O{\ensuremath{\eta} } m }{\ensuremath{\pi_{#3}  }}
\NewDocumentCommand{\snIrrepDim}{ O{N} O{M} }{\ensuremath{ m_{#2, #1} }}

\NewDocumentCommand{\PhiChannel}{ O{\ensuremath{\eta}} O{M} O{N} }{\ensuremath{ \Phi_{#2 \to #3}^{#1} }}
\NewDocumentCommand{\unnormalPhiChannel}{ O{\ensuremath{\eta}} O{M} O{N} }{\ensuremath{ \Psi_{#2 \to #3}^{#1} }}

\NewDocumentCommand{\cloneChannel}{ O{s} O{N} }{\ensuremath{ C_{#1 \to #2} }}
\NewDocumentCommand{\traceChannel}{ O{s} O{M}}{\ensuremath{ L_{#2 \to #1} }}

\newcommand{\dicke}[2]{\ket{D_{#2}^{#1}}}
\newcommand{\cP}{\mathcal{P}}
\newcommand{\cQ}{\mathcal{Q}}
\newcommand{\one}{\mathbb{I}}

\crefname{theorem}{theorem}{theorems}
\Crefname{theorem}{Theorem}{Theorems}

\crefname{lemma}{lemma}{lemmas}
\Crefname{lemma}{Lemma}{Lemmas}

\crefname{proposition}{proposition}{propositions}
\Crefname{proposition}{Proposition}{Propositions}

\crefname{corollary}{corollary}{corollaries}
\Crefname{corollary}{Corollary}{Corollaries}

\crefname{conjecture}{conjecture}{conjectures}
\Crefname{conjecture}{Conjecture}{Conjectures}

\crefname{definition}{definition}{definitions}
\Crefname{definition}{Definition}{Definitions}

\crefformat{enumi}{(#2#1#3)}
\crefmultiformat{enumi}{(#2#1#3)}{ and~(#2#1#3)}{, (#2#1#3)}{ and~(#2#1#3)}
\crefrangeformat{enumi}{(#3#1#4)-(#5#2#6)}

\definecolor{dgreen}{HTML}{006600}
\definecolor{lgreen}{HTML}{B3FFB3}

\newcommand{\cD}{\mathcal{D}}

\newcommand{\cL}{\mathcal{L}}

\newcommand{\cN}{\mathcal{N}}

\newcommand{\cT}{\mathcal{T}}

\DeclareMathOperator{\id}{id}

\def\Binom{\mathrm{Binom}}
\def\NegBinom{\mathrm{NegBinom}}
\def\Hypergeo{\mathrm{HyperGeo}}
\def\NegHypergeo{\mathrm{NegHyperGeo}}

\def\Gaussian{\Gamma}
\def\Gain{\mathcal{A}}
\def\Loss{\mathcal{L}}

\def\Sym{\mathrm{Sym}}

\def\K{\tilde K}

\def\E{\mathbb{E}}

\def\Id{\mathrm{id}}
\def\id{\mathbb{I}}

\newcommand{\KDiaNorm}[2][K]{ \norm{#2}_{\diamond, #1}}

\newcommand{\diver}[2]{D \hspace{-0.7mm} \left( #1 \middle\| #2\right)}

\newcommand\Jloss[1][s] {J_{#1}^{\text{loss}} }

\newcommand{\be}{\begin{equation}}
\newcommand{\ee}{\end{equation}}

\newcommand{\C}{\mathbb{C}}

\newcommand{\bea}{\begin{eqnarray}}
\newcommand{\eea}{\end{eqnarray}}
\newcommand{\beas}{\begin{eqnarray*}}
\newcommand{\eeas}{\end{eqnarray*}}

\NewDocumentCommand{\glIrrep}{ O{N} O{M} }{\ensuremath{ \cQ_{#2, #1} }}
\NewDocumentCommand{\snIrrep}{ O{N} O{M} }{\ensuremath{ \cP_{#2, #1} }}

\newtheoremstyle{suppplain}
  {6pt}{6pt}{\itshape}{}{\bfseries}{.}{.5em}
  {\thmname{#1}\thmnumber{ (#2)}\thmnote{ \textnormal{(#3)}}}

\newtheoremstyle{suppdefinition}
  {6pt}{6pt}{\normalfont}{}{\bfseries}{.}{.5em}
  {\thmname{#1}\thmnumber{ (#2)}\thmnote{ \textnormal{(#3)}}}

\theoremstyle{suppplain}
\newtheorem{supptheorem}{Theorem}[section]
\newtheorem{supplemma}[supptheorem]{Lemma}
\newtheorem{suppproposition}[supptheorem]{Proposition}
\newtheorem{suppcorollary}[supptheorem]{Corollary}
\newtheorem{suppconjecture}[supptheorem]{Conjecture}

\theoremstyle{suppdefinition}
\newtheorem{suppdefinition}[supptheorem]{Definition}
\newtheorem*{suppremark}{Remark}

\crefname{supptheorem}{theorem}{theorems}
\Crefname{supptheorem}{Theorem}{Theorems}
\crefname{supplemma}{lemma}{lemmas}
\Crefname{supplemma}{Lemma}{Lemmas}
\crefname{suppproposition}{proposition}{propositions}
\Crefname{suppproposition}{Proposition}{Propositions}
\crefname{suppcorollary}{corollary}{corollaries}
\Crefname{suppcorollary}{Corollary}{Corollaries}
\crefname{suppconjecture}{conjecture}{conjectures}
\Crefname{suppconjecture}{Conjecture}{Conjectures}
\crefname{suppdefinition}{definition}{definitions}
\Crefname{suppdefinition}{Definition}{Definitions}
\begin{document}

\preprint{APS/123-QED}

\title{A depolarizing choir sings in Gaussian harmony}  

\author{Rabsan Galib Ahmed$^{1}$}
\email{rgahmed@uwaterloo.ca}
\author{Sujeet Bhalerao$^{2}$}%
\email{sgb4@illinois.edu}
\author{Sungjai Lee$^{1}$}%
\email{sungjai.lee@uwaterloo.ca}
\author{Felix~Leditzky$^{2}$}%
\email{leditzky@illinois.edu}
\author{Debbie Leung$^{1,3}$}%
\email{wcleung@uwaterloo.ca}
\author{Luke Schaeffer$^{1}$}%
\email{lschaeffer@uwaterloo.ca}
\author{Graeme Smith$^{1}$}%
\email{graeme.smith@uwaterloo.ca}
\affiliation{$^1$University of Waterloo, Waterloo, Ontario, Canada.\\
$^2$University of Illinois Urbana-Champaign, Urbana, Illinois, USA.\\
$^3$Perimeter Institute for Theoretical Physics, Waterloo, Ontario, Canada.
}%


\date{\today}

\begin{abstract}
We study the noise threshold for positive quantum capacity for the qubit depolarizing channel. We explore analytically the action of the qubit depolarizing channel on the symmetric subspaces of the input qubits, in the limit of asymptotically many uses of the channel. We observe the emergence of a bosonic Gaussian channel.  Furthermore, the codes previously developed for the depolarizing channel can be translated to codes for the emergent Gaussian channel, and it is easier to further optimize these codes for the simpler emergent Gaussian channel.  Translating these codes back to the depolarizing channel leads to extremely good input states for the coherent information of the depolarizing channel producing new lower bounds on the noise threshold for positive capacity. In addition to improved lower bounds on the threshold, this newly found link between depolarizing noise and Gaussian channels offers a novel perspective contributing to our understanding of these symmetric codes.
\end{abstract}

\maketitle

A fundamental question in quantum information theory addresses the maximum achievable rate of reliable quantum communication over asymptotically many uses of a given channel, known as the quantum capacity $Q$ of the channel. In contrast to the analogous quantity for a classical channel, for which there is a single-letter optimization to evaluate the capacity \cite{shannon}, the quantum capacity of a quantum channel requires an optimization of \textit{coherent information} over a potentially unbounded number of channel uses \cite{lloyd1997capacity,shor2002quantum,devetak2005private}  
 making the quantification difficult. Even certifying when a channel has positive capacity seems challenging \cite{cubitt2015unbounded}. 

A classical binary symmetric channel flips a classical bit with some probability.  In the quantum setting, this flip can be applied along three bases, given by the Pauli-$X$, $Y$, or $Z$ operators on a qubit.
The qubit depolarizing channel applies each of the three errors $X$, $Y$, or $Z$ with equal probability $p/3$, and it leaves the input qubit state $\rho$ unchanged with probability $1-p$.  The channel is given by $\mathcal{D}_p(\rho) = (1-p)\rho + (p/3)(X\rho X+ Y\rho Y+Z\rho Z)$. Equivalently, the depolarizing channel acts as 
\begin{equation}
\mathcal{D}_\eta(\rho) = \eta \rho + (1-\eta) \frac{\mathbb{I}}{2}\Tr\rho , 
\end{equation}
where the input state $\rho$ is unchanged with probability $\eta = (1-4p/3)$ (also called the 
transmissivity) and replaced by a maximally mixed state otherwise.   

Depolarizing noise is a fundamental and well-studied error model in quantum information processing, yet its quantum capacity is unknown.
In fact, the supremum value of the noise parameter $p$ for which the quantum capacity of $\mathcal{D}_p$ is positive, called the threshold $p_{th}$, is not known exactly, except for lower bounds $p_{th} \gtrsim 0.1940$ obtained by explicit inner codes (which are then concatenated with a random code) \cite{divincenzo1998capacity,shor1996quantumerrorcorrectingcodesneed,smith2006degenerate,fern2008lower,bausch2021error,agarwal2026enhanced}, and an upper bound of $p_{th} \leq 1/4$ obtained by a no-cloning argument~\cite{bruss1998optimal}. 

The state-of-art lower bound $p_{th} \gtrsim 0.1940$ mentioned above comes from the recent 
work \cite{agarwal2026enhanced}.  The authors reported a substantial improvement using a representation theoretic framework developed in \cite{bhalerao2025improvingquantumcommunicationrates}; they restricted the input state of $M$ copies of $\mathcal{D}_{p}$ to the symmetric subspace of $M$ qubits and further optimized this ansatz over low-dimensional code spaces.  However, computational requirements limit reliable results to $M \approx 45$, despite an upward trend of the lower bound with $M$.  Furthermore, there is no compelling explanation for the numerically optimal solutions. Small numerical improvements have been found by LLMs subsequently \cite{Anthony,Artus}.

In this work, we explore analytically the action of the qubit depolarizing channel with transmissivity $\eta$ on the symmetric subspaces of the input qubits, in the limit of asymptotically many uses of the channel. We observe the \emph{emergence} of a Gaussian thermal attenuator of transmissivity $\eta$ and environmental mean photon number $(2\eta)^{-1}$. This emergence is qualitatively similar to the local asymptotic normality~\cite{Gu__2006}. It is also closely related to the bosonic lifting of quantum blurring developed in~\cite{Lami_2025} and further studied in~\cite{grinko2026samplecomplexityquantumresource}. Furthermore, the codes previously developed for the depolarizing channel can be translated to codes for the emergent Gaussian channel, and it is easier to further optimize these codes for the simpler emergent Gaussian channel. Translating these codes back to the depolarizing channel leads to extremely good input states for the coherent information of the depolarizing channel, producing a new improved lower bound $p_{th}\gtrsim 0.2029$. In addition to improved thresholds, this newly found link between depolarizing noise and Gaussian thermal attenuation offers a novel perspective contributing to our understanding of these symmetric codes. 

As an aside, this study also enables progress for the noise threshold problem for the Gaussian channel with this parameter relation.  Our work provided several codes that outperform any Gaussian code for our Gaussian channel. The first such code was found only very recently in \cite{mele2026bosonicthermal} for similar Gaussian channels, but for a different parameter range.

\textit{Preliminaries ---}
For a Hilbert space $\mathcal{H}$ we denote by $\mathcal{L}(\mathcal{H})$
the algebra of linear operators on $\mathcal{H}$.
The coherent information of a state $\rho_A$ with purification $\ket{\psi_{RA}}$ through a
channel $\mathcal{N}$ is $I_c(\rho,\mathcal{N})=S(\mathcal{N}(\rho))-S\bigl((\mathrm{id}_R\otimes\mathcal{N})(\psi)\bigr).$
The quantum capacity is then given by the regularized expression \cite{lloyd1997capacity,shor2002quantum,devetak2005private}
$$Q(\mathcal{N})=\lim_{n\to\infty}\frac1n\max_\rho I_c(\rho,\mathcal{N}^{\otimes n}).$$
For any state $\rho$ on $\Sym^M(\mathbb{C}^2)$ one has
$Q(\cD_\eta)\ge\frac1M I_c(\rho,D_\eta^{\otimes M}),$ thus finding a state $\rho$ on $\Sym^M(\mathbb{C}^2)$ for which $I_c(\rho,\cD_\eta^{\otimes M}) > 0$ certifies $p = \frac{3}{4}(1-\eta)$ as a lower bound for the quantum capacity threshold of the channel.

The groups $\mathrm{GL}(2)$ and $S_M$ act on $(\mathbb{C}^2)^{\otimes M}$
by $A \mapsto A^{\otimes M}$ and by permuting tensor factors, respectively. These actions commute, and Schur-Weyl duality states that these two representations span each other's commutants in $\cL((\mathbb{C}^2)^{\otimes M})$, yielding a decomposition of the representation space as \cite{goodman2009symmetry,etingof2011introduction,fulton2013representation} 
\begin{equation}
  (\mathbb{C}^2)^{\otimes M} \cong \bigoplus_{N} \mathcal{Q}_{M,N} \otimes \mathcal{P}_{M,N},
  \label{eq:schur-weyl-duality}
\end{equation}
where $N \in \{M, M{-}2, \dots\}$ labels the Young diagram 
$\big(\tfrac{M{+}N}{2}, \tfrac{M{-}N}{2}\big)$ (we will call $N$ the spin sector, or simply sector), $\tfrac{N}{2}$ is called the \emph{total spin} or \emph{total angular momentum}, where $N=0,2,\dots,M$ for even $M$ and $N=1,3,\dots, M$ for odd $M$, 
$\mathcal{Q}_{M,N}$ is a $\mathrm{GL}(2)$-irreducible representation (irrep),
and $\mathcal{P}_{M,N}$ is an $S_M$-irrep of dimension 
\begin{align}\label{eq: m_MN}\snIrrepDim  = \binom{M}{\frac{M-N}{2}} - \binom{M}{\frac{M-N}{2}-1}.
\end{align}
By Schur's Lemma,
any permutation-invariant operator has the form $\bigoplus_N X_N \otimes \one_{M,N}$ with respect to \eqref{eq:schur-weyl-duality}, where $\one_{M,N}$ denotes the identity on $\cP_{M,N}$.
Independent and identically distributed (IID) operators $A^{\otimes M}$ can be expressed explicitly as \cite{goodman2009symmetry}
\begin{align} 
A^{\otimes M} & \cong \bigoplus_N (\det A)^{\frac{M-N}{2}} S_N(A) \otimes \mathbb{I}_{M,N} \,, 
\label{eq:schurlemma} 
\end{align} 
where $S_N(A) = P_N A^{\otimes N} P_N$ 
is the restriction of $A^{\otimes N}$ to $\Sym^N(\mathbb{C}^2)$, and $P_N$ is the projection operator onto $\Sym^N(\mathbb{C}^2)$. 
When $N = M$, $\mathcal{Q}_{M,N}$ is the symmetric subspace, with Dicke basis 
\begin{equation}
\dicke{M}{k} = \binom{M}{k}^{-1/2}\sum_{\mathrm{wt}(x)=k}\ket{x}\,.
\end{equation}
We refer the interested readers to \cite{goodman2009symmetry,fulton2013representation} for details.  

We use two maps between symmetric subspaces.  The first is the loss channel
$$\traceChannel(X)=\tr_{M-s}X \,,$$ which maps $\mathcal{L}(\Sym^M(\mathbb{C}^2))$ to
$\mathcal{L}(\Sym^s(\mathbb{C}^2))$ and satisfies ${\traceChannel[s][r] \circ \traceChannel[r][M] = \traceChannel[s][M]}$. 
The second map is the optimal universal cloner \cite{Werner1998}; 
in terms of the projector $P_N$ onto $\Sym^N(\C^2)$, it is 
\begin{equation}
\cloneChannel (Y)=\frac{s+1}{N+1}\,P_N\big(Y\otimes \mathbb{I}_2^{\otimes(N-s)}\big)P_N . 
\end{equation}

Channels in optical quantum communication are best described in terms of bosonic systems. A single bosonic mode corresponds to a separable Hilbert space, spanned by the Fock states: $\bigl\{\ket{k}: k\in \{0\}\cup\mathbb{N} \bigr\}$. Physically, $k$ corresponds to the number of photons present in the mode. Among several examples of bosonic channels, the two most relevant to this work are the single mode pure-loss attenuation channel, $\mathcal{L}_T$ and the quantum-limited amplification channel, $\mathcal{A}_{G}$~\cite{Holevo2001}. Mathematically, a pure-loss attenuation channel is realized by mixing the input mode with some environmental vacuum mode on a beam splitter of transmissivity $0\leq T \leq 1$, followed by tracing out the environment~\cite{eisert2005}. Physically, it arises when a fraction $1-T$ of a signal is absorbed. A quantum-limited amplification channel is mathematically realized by jointly acting on the input mode and some environmental vacuum mode with a two-mode squeezer with gain $G\ge 1$, followed by tracing out the environment~\cite{eisert2005}. Physically, the signal is amplified by a factor $G$ with some added noise.

The composition of two Gaussian channels is also a Gaussian channel. In particular, for $G T<1$, the composition $\mathcal{A}_G\circ \mathcal{L}_T$ is a Gaussian thermal attenuation channel with transmissivity $GT$ and an environmental mean photon number $(G-1)/(1-GT)$. Physically, a thermal attenuator is the same as a pure-loss attenuation with the environment being prepared in a thermal state instead of a vacuum. 


\textit{Main results--} Consider a state $\rho \in {\cal L}(\Sym^M(\C^2))$
which is permutation-invariant.  
After the application of $M$ IID~depolarizing channels on each individual qubit, which are permutation-covariant, and thus preserving the permutation invariance of any input state, 
it follows that the output $\mathcal{D}_{\eta}^{\otimes M}(\rho)$ is also permutation-invariant, 
and decomposes into different spin sectors labeled by $N$  
and can thus be written as
\begin{align} 
\mathcal{D}_\eta^{\otimes M}(\rho) = \bigoplus_{N} \unnormalPhiChannel(\rho) \otimes \frac{\mathbb{I}_{M,N}}{\snIrrepDim}
\end{align}
for some linear completely positive maps
$\unnormalPhiChannel$. We have the following theorem \footnote{Note that throughout the manuscript we are suppressing the dependence on $\eta$ of scalar quantities like $p_N$ to increase readability.
Furthermore, we often also suppress dependence on $M$, e.g., for $p_N$ defined in \eqref{eq: p_N} and $\pi_s$ defined in \eqref{eq: pi_s}.}.  
\begin{theorem}\label{thm:decomp-main}
    For any operator $X\in \cL(\Sym^M(\C^2))$,
    \begin{align}
        \mathcal{D}_\eta^{\otimes M}(X) = \bigoplus_{N}\; \pDist{N}\;\PhiChannel(X) \otimes \frac{\mathbb{I}_{M,N}}{\snIrrepDim}
    \end{align}
    for some quantum channels $\PhiChannel\colon {\cal L}(\Sym^M(\C^2))$ $\to$ ${\cal L}(\Sym^N(\C^2))$ and a fixed probability distribution
    $\{\pDist{N}\}_N$:
    \begin{align}\label{eq: p_N}
        &\hspace*{-2ex} \pDist{N}  = \snIrrepDim \left(\frac{1-\eta^2}{4}\right)^{\hspace*{-1ex}\frac{M-N}{2}} \hspace*{-0.5ex} \frac{(1+\eta)^{N+1}-(1-\eta)^{N+1}}{2^{N+1}\eta}, \hspace*{-1ex}
    \end{align}
    with $\snIrrepDim=\dim\cP_{M,N}$ defined in \eqref{eq: m_MN}.
\end{theorem}

\begin{proof}[Proof sketch] 
Using the unitary covariance of the depolarizing channel, one can show that, for any operator $X\in {\cal L}(\Sym^M(\C^2))$, 
$\tr \unnormalPhiChannel(X)  = \pDist{N} \tr X$ where $\pDist{N}$ is independent of $X$ (see Lemma~\ref{lem:sector-probability-input-ind} in Supplemental), and $\left\{\pDist{N} \right\}_N$ forms a probability distribution.
Furthermore, as the $N/2$-spin sector is $(N+1)$-dimensional, the corresponding $\GL(2)$-irreps are isomorphic to the symmetric subspace of $N$ qubits.

    Since $\pDist{N}$ does not depend on the input state, we may evaluate it at  $\ketbra{0}{0}^{\otimes M}$. The channel output is $A^{\otimes M}$ with $A = \diag(a,b)$, where $a = \frac{1+\eta}{2}$ and $b = \frac{1-\eta}{2}$. 
    By the explicit expression in Eq.~\eqref{eq:schurlemma},  $A^{\otimes M}$ acts on the $N$-th sector as $(\det A)^{\frac{M-N}{2}} S_N(A) \otimes \mathbb{I}_{M,N}$. 
    Taking the trace gives
\begin{align*}
\pDist{N} &= \snIrrepDim\,(ab)^{\frac{M-N}{2}}\,q_N,\\
q_N &\coloneqq \Tr\,S_N(A) = \frac{a^{N+1}-b^{N+1}}{\eta},
\end{align*}
which gives \eqref{eq: p_N} with $a = \frac{1+\eta}{2}$ and $b = \frac{1-\eta}{2}$.
\end{proof}

Next we demonstrate that each of these conditional channels $\PhiChannel$ is equivalent to randomly losing $M-s$ qubits and optimally cloning the state of the $s$ qubits to $N$ qubits \footnote{Note that as the state of the $M$ qubits we begin with is permutation-invariant, it is irrelevant which $s$ qubits are retained}.

\begin{theorem} The conditional channels satisfy
    \begin{align}\label{eq:irrep-channel-decomp}
        \PhiChannel
        =
        \sum_{s=0}^{N}
        \piDist{s}\,
        \cloneChannel
        \circ
        \traceChannel,
    \end{align}
where the probability distribution over $s$, the number of intermediate retained qubits, is given by 
\begin{align}\label{eq: pi_s}
    \piDist{s} = {N+1 \choose s+1}\frac{(2\eta)^{s+1}(1-\eta)^{N-s}}{(1+\eta )^{N+1}- (1-\eta)^{N+1}},
\end{align}
and $\traceChannel \colon {\cal L}(\Sym^M(\C^2))$ $\to$ ${\cal L}(\Sym^s(\C^2))$ and $\cloneChannel \colon {\cal L}(\Sym^s(\C^2))$ $\to$ ${\cal L}(\Sym^N(\C^2))$ are the loss and cloning channels, respectively. 
\end{theorem}
\begin{proof}[Proof sketch]
    We first show that $\PhiChannel =\Theta_N\circ \traceChannel[N]$, where
\begin{equation*}
\Theta_N(X) = q_N^{-1}\,P_N\,\mathcal{D}_\eta^{\otimes N}(X)\,P_N
\end{equation*}
and $P_N$ is the projector onto $\Sym^N(\C^2)$. By
linearity it suffices to check this equality on $\ketbra{u}{u}^{\otimes M}$, since these operators span ${\cal L}(\Sym^M(\C^2))$ (see \cite[Section 1.1]{harrow2013church}). The output of such an input
is $A_u^{\otimes M}$, where $A_u := \mathcal{D}_\eta(\ketbra{u}{u})$
has eigenvalues $a$ and $b$. Repeating the sector computation from the
proof of Theorem \ref{thm:decomp-main} with $A$ replaced by $A_u$ gives
$\PhiChannel (\ketbra{u}{u}^{\otimes M}) = S_N(A_u)/q_N$,
which equals $\Theta_N \circ \traceChannel[N] (\ketbra{u}{u}^{\otimes M})$. Next, write $\mathcal{D}_\eta=\eta\,\mathrm{id}+b\, \mathcal{T}$ with
$\mathcal{T}(\cdot)=\tr(\cdot)\,\mathbb{I}_2$. 
Expanding
$\mathcal{D}_\eta^{\otimes N}$ and using the permutation invariance
of $X$,
\begin{multline}
P_N\mathcal{D}_\eta^{\otimes N}(X)P_N
\\= \sum_{s=0}^{N}\binom{N+1}{s+1}\eta^s b^{N-s}\,
\cloneChannel \big(\traceChannel[s][N](X)\big),
\end{multline}
where
$\cloneChannel (Y)=\tfrac{s+1}{N+1}P_N\big(Y\otimes\mathbb{I}_2^{\otimes(N-s)}\big)P_N$
is the optimal cloner of \cite{Werner1998}. Dividing by $q_N$ gives
the weights in Eq.~\eqref{eq: pi_s}. Finally, composing with $\traceChannel[N]$ and using $\traceChannel[s][N] \circ \traceChannel[N][M]  =\traceChannel[s][M]$ gives
Eq.~\eqref{eq:irrep-channel-decomp}.
\end{proof}

Now, we begin to analyze the asymptotic behavior of $\mathcal{D}_\eta^{\otimes M}$ as $M$, the number of the input qubits, grows. First, we note that $\pDist{N}$ exhibits a concentration phenomenon for the typical spin-sector, leading to $N/M$ getting concentrated around $\eta$ in the large $M$ limit.
\begin{lemma}\label{lem: con.in.N}
    For any $\delta_1 >0$ and every $M$,
    \begin{align}
        \mathbb{P}_{\pDist{N}}\left(\abs{\frac{N}{M} - \eta}>\delta_1\right) \leq \frac{1+\eta}{\eta}\;\exp\left(-\frac{M\delta_1^2}{2}\right).
    \end{align}
\end{lemma}
The proof sketch is included in the Appendix~\ref{appsec:lemma3proof}. From the Lemma above, we see that as $M$ grows, the typical $N$ also grows. This leads to a second concentration phenomenon for $\piDist{s}$ for such typical sectors. We see that $s/N$ concentrates around $2\eta/(1+\eta)$ in the large $N$ limit.
\begin{lemma}\label{lem:con.in.s}
    For any $1\geq \delta_2 >0$, and $c=e^{4\frac{1-\eta}{1+\eta}}$, we have that for 
    sufficiently large $N$,
    \begin{align}\label{eq:con.in.s-main}
\mathbb{P}_{\piDist{s}}\left(\abs{\frac{s}{N} - \frac{2\eta}{1+\eta}}>\delta_2\right) \leq \frac{1+\eta}{\eta}c\;\exp\left(-\frac{2N^2\delta_2^2}{(N+1)}\right).
    \end{align}
\end{lemma}
\begin{proof}[Proof sketch]
    As elaborated in the supplemental material, we can write
    \begin{align}
        \piDist{s} \leq \frac{1+\eta}{2\eta}\; p_{\mathrm{Bin}}(Y_{N+1} = s+1).
    \end{align}
    Here we have introduced another binomially distributed random variable $Y_{N+1}$ similar to the proof of Lemma~\ref{lem: con.in.N}, however with the success probability $2\eta/(1+\eta)$. We then obtain the inequality~\eqref{eq:con.in.s-main} in a similar fashion.
\end{proof}
    
With these two concentrations in place, we now turn our attention to spin-sectors given by $N$ which scale with $M$. We find that on these sectors, within a fixed excitation cutoff, $K$, each conditional channel asymptotically acts as a bosonic pure-loss attenuation channel, with transmissivity $2\eta(N/M)/(1+\eta)$, followed by a quantum-limited amplifier channel, with gain $(1+\eta)/(2\eta)$. In particular, this convergence happens for the typical conditional channels. Here, a natural mapping between the Dicke states and the bosonic Fock states has been assumed. To quantify the convergence we use the diamond norm of the restriction of a linear map $\Phi$ to the first $K$ excitations: $\norm{\Phi}_{\diamond,K}:= \norm{\Phi\vert_{\mathcal{F}_K}}_{\diamond}$, where ${\mathcal{F}_K =\mathrm{Span}\{\ket{k}: k=0,1,\dots K\}}$. As the input space of this restricted channel is finite-dimensional, this constitutes a valid norm~\cite{Paulsen_2003}. We prove the following result.
\begin{theorem}\label{thm:diamond-dist-bound-main} For a fixed excitation cutoff, $K\leq M$, the typical sectors, $N/M \in \left[\eta -\sqrt{\frac{2\log M}{M}}, \eta +\sqrt{\frac{2\log M}{M}}\right]$, satisfy
    \begin{align}\label{eq:diamond-bound-thm5}
        \norm{\PhiChannel - \mathcal{A}_{\frac{1+\eta}{2\eta}}\circ \mathcal{L}_{\frac{2\eta}{1+\eta}\frac{N}{M}}}_{\diamond, K} \leq \varepsilon_{M,K},
    \end{align}
    where $\varepsilon_{M,K}\to 0$ as $M\to \infty$.
\end{theorem}
\begin{proof}[Proof sketch]
    In Theorem~\ref{thm:diamond-dist-bound} of the Supplemental material we show that for every fixed ratio $0<N/M<1$ and a fixed cutoff $K\leq M$, the left hand side of Eq.~\eqref{eq:diamond-bound-thm5} is upper-bounded by $\varepsilon_{M,K,(N/M)}$ which approaches zero as $M$ approaches infinity. The proof can be understood simply as the following. On every excitation state $\ket{k}$ with $k\leq K$, the action of the Kraus operators of $\cloneChannel$ and $\traceChannel$ resembles the negative hypergeometric distribution and the hypergeometric distribution respectively. In the large population limit, these distributions approach a negative binomial distribution with success probability $(s+1)/(N+1)$ and a binomial distribution with success probability $s/M$, respectively, for fixed $K$~\cite{Kinney2009-oq}. These are precisely the coefficients of $\mathcal{A}_{(N+1)/(s+1)}$ and $\mathcal{L}_{s/M}$ respectively on the Fock state $\ket{k}$. 

    Moreover, as $M$ grows (and thus $N$ since $N/M>0$), Lemma~\ref{lem:con.in.s} implies a concentration of $s/N$ around $\frac{2\eta}{1+\eta}$. Hence within the typical interval for $s$, $\cloneChannel$ approaches $\mathcal{A}_{\frac{1+\eta}{2\eta}}$ and $\traceChannel$ approaches $\mathcal{L}_{\frac{2\eta}{1+\eta}\frac{N}{M}}$. 

    Finally we conclude that as the convergence occurs for all fixed positive ratios $N/M$, it particularly occurs for the typical interval for $N/M$ and a suitable upper bound on the diamond distance, $\varepsilon_{M,K}$, is presented in the Supplemental material. We see explicitly that for fixed $K$, $\varepsilon_{M,K}\to 0$ as $M\to \infty$.
\end{proof}
    
The convergence in diamond norm between the typical conditional channels and a corresponding Gaussian channel, under the fixed cutoff $K$ indeed implies convergence in coherent information evaluated on states supported within the cutoff. Therefore, we can chain these convergences and conclude:
\begin{theorem}\label{thm:coherent-information} For every $\rho\in {\cal L}(\mathrm{Sym}^M(\mathbb{C}^2))$ supported within a fixed Dicke excitation cutoff (equivalently Fock cutoff) given by $K$, 
    \begin{align}
        \abs{I_c(\rho,\mathcal{D}^{\otimes M}_\eta) - I_c(\rho, \mathcal{A}_{\frac{1+\eta}{2\eta}}\circ \mathcal{L}_{\frac{2\eta^2}{1+\eta}})} \leq \Delta_{M,K},
    \end{align}
    where $\Delta_{M,K}\to 0$ as $M\to \infty$.
\end{theorem}
\begin{proof}[Proof sketch]
   Using the block decomposition form of the output of $\mathcal{D}_{\eta}^{\otimes M}$ as obtained from Theorem~\ref{thm:decomp-main}, we can write $I_c(\rho,\mathcal{D}^{\otimes M}_\eta) = \sum_{N}p_N I_c(\rho,\Phi^\eta_{M\to N})$. As $M$ grows, Lemma~\ref{lem: con.in.N} implies that $N/M$ concentrates around $\eta$. Within the typical interval given by $\delta_M = \sqrt{(2\log M)/M}$, we apply Theorem~\ref{thm:diamond-dist-bound-main} and the continuity of conditional entropy~\cite{Winter_2016,Alicki_2004} to upper-bound $\big|I_c(\rho,\Phi^\eta_{M\to N}) - I_c(\rho, \mathcal{A}_{\frac{1+\eta}{2\eta}}\circ \mathcal{L}_{\frac{2\eta}{1+\eta}\frac{N}{M}})\big|$ in terms of $\varepsilon_{M,K}$. Moreover, we bound $\big|I_c(\rho,\mathcal{A}_{\frac{1+\eta}{2\eta}}\circ \mathcal{L}_{\frac{2\eta}{1+\eta}\frac{N}{M}}) - I_c(\rho, \mathcal{A}_{\frac{1+\eta}{2\eta}}\circ \mathcal{L}_{\frac{2\eta^2}{1+\eta}})\big|$ in the typical interval of $N/M$ in terms of $\varepsilon^{\mathrm{att}}_{M,K}$, an upper bound on the diamond distance between the respective attenuation channels within the typical interval (see Lemma~\ref{lem:typical-gaussian-channels-convergence} in Supplemental material). Finally, using the triangle inequality we obtain a bound on $\big|I_c(\rho,\Phi^\eta_{M\to N}) - I_c(\rho, \mathcal{A}_{\frac{1+\eta}{2\eta}}\circ \mathcal{L}_{\frac{2\eta^2}{1+\eta}})\big|$ in the typical interval for $N/M$.

   Outside the typical interval, we replace the absolute value of the difference in coherent information by its upper bound $2\log (K+1)$. Combining everything, we finally obtain the stated bound of $\Delta_{M,K}$, which approaches $0$ as $M\to \infty$. 
\end{proof}


\def\bestEta{0.7294452}
\def\bestP{0.2029161}
\def\bestPriori{0.4506152}
\def\bestCon{56.47}
\def\bestK{220}
\def\bestD{320}
\def\bestIc{10^{-12}}

\textit{Improved lower bound for the quantum capacity threshold ---} Using the convergence as stated in Theorem~\ref{thm:coherent-information}, we can conclude that if the Gaussian channel has a positive coherent information on states supported on a fixed finite number of excitations, a sufficiently large number of uses of the corresponding depolarizing channel also has positive coherent information. To compute the coherent information of the Gaussian channel, we adopt a method similar to \cite{mele2026bosonicthermal}; we truncate the output of the Gaussian channel to a finite cutoff $D$ and collapse all higher excitation levels to $\ket{\perp}$ orthogonal to the lower $D$ levels. Then, by data processing, we can lower-bound the coherent information of the full Gaussian output in terms of the truncated coherent information. Finally, we verify that such a lower bound is positive.

Fix ${\eta = \bestEta}$, which corresponds to ${p = \bestP}$. The emergent Gaussian channel is then $\mathcal{G}_{\eta} = \mathcal{A}_G \circ \mathcal{L}_T$ with $T = 2\eta^2/(1+\eta)$ and $G = (1+\eta)/2\eta$. This is a thermal attenuator with transmissivity $\eta$ and environmental mean photon number $1/(2\eta)$.

As input we take a rank-two state of the form 
\begin{equation}
    \rho = q \ketbra{\psi_0}{\psi_0} + (1-q) \ketbra{\psi_1}{\psi_1}
\end{equation}
where $\ket{\psi_0}$ is supported on Fock states $\ket{n}$ with $n \equiv 0 \pmod 3$ and $\ket{\psi_1}$ is supported on $n \equiv 1 \pmod 3$~\cite{githublink}. Both codewords have real amplitudes and are supported on $n \leq \bestK$ with $ q \approx \bestPriori.$ We find $\rho$ by optimizing $I_c(\rho,\mathcal{G}_\eta)$ directly. This is an optimization over two codewords of a single bosonic mode, and is easier than the corresponding optimization over $M$ qubits for the depolarizing channel. The amplitudes of both codewords have a nearly Gaussian envelope, with mean photon number close to $\bestCon$. The output of the channel is truncated at $D = \bestD$ photons. At $\eta = \bestEta$ we obtain $I_c(\rho,\mathcal{G}_{\eta}) \approx \bestIc$. Since $I_c(\rho,\mathcal{G}_{\eta})> 0,$ Theorem~\ref{thm:coherent-information} gives $I_c(\rho, \mathcal{D}_{\eta}^{\otimes M}) \ge I_c(\rho, \mathcal{G}_\eta) - \Delta_{M,K} >0$ for all sufficiently large $M$. Hence $Q(\mathcal{D}_\eta) \ge \frac{1}{M} I_c(\rho, \mathcal{D}_{\eta}^{\otimes M}) > 0.$ We therefore conclude from our work that ${p_{\mathrm{th}} \ge \bestP}.$ This provides the best-known lower bound on the threshold for the quantum capacity of the qubit depolarizing channel.  
In comparison, the previous best-known lower bound from \cite{agarwal2026enhanced} was $p_{\mathrm{th}} \ge 0.19397$ obtained recently, which improved upon the lower bound of $0.19128$ obtained by \cite{fern2008lower} in 2008.  The first lower bound from \cite{shor1996quantumerrorcorrectingcodesneed} in 1996 was $0.19036$, and the hashing bound is $0.18929$.  A timeline on the improvement of the lower bound of the threshold is as follows: 
\begin{figure}[H]
    \centering 
   \includegraphics[width=0.48\textwidth]{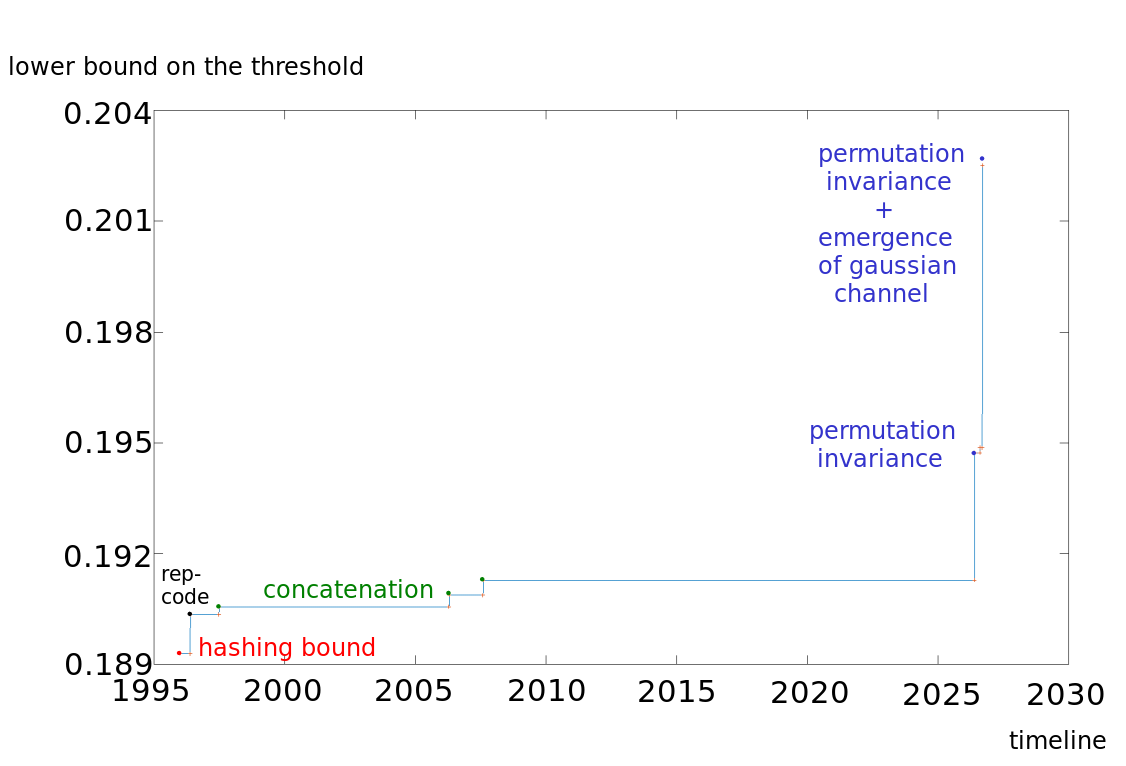}
    \caption{Best known lower bounds for the threshold as a function of time (in calendar years) and the coding / analysis giving rise to them.}
\end{figure}

We finally note that smaller instances of such mod $3$ codes are known in the bosonic codes literature~\cite{Michael_2016}.

\textit{Acknowledgments ---} SB and FL are supported by National Science Foundation Grant No.~2442410. 
GS, RGA, and SL are supported under NSERC-NSF alliance grant ALLRP-586858-2023 and NSERC Discovery grant
RGPIN-2025-02094.
RGA acknowledges the support of the Institute for Quantum Computing and the Mike and Ophelia Lazaridis Graduate Fellowship. LS is supported under NSERC RGPIN-2025-04875.  DL is supported under NSERC RGPIN-2024-03823 and an NSERC Alliance Consortia Quantum grants (ALLRP 578455-22). 
We acknowledge discussions with Lauritz van Luijk about the diamond norm and with Zhiyao Wang about bosonic codes. The authors acknowledge the hospitality of the Banff International Research Station for Mathematical Innovation and Discovery in Banff, Alberta, Canada during the workshop ``Additivity Problems in Quantum and Classical Information Theory [26w5621]'', July 12-17, 2026, during which parts of this work discussed. 

\textit{Data availability statement---} The data relevant to this work can be found in the following github link~\cite{githublink}.

\textit{AI statement---} The Gaussian comb structure of the numerical optimizers from~\cite{agarwal2026enhanced} was observed without the assistance of AI.  The emergence of Gaussianity in the action of the channel on the symmetric space was suggested by AI. Theorem 1 is an application of Schur-Weyl duality known to us without AI. The proof of Theorem 2 was suggested by AI and carefully digested and confirmed by the authors. The proofs of Lemmas 3 and 4 were created by the authors after AI suggested unbearable proofs of them. Theorems 5 and 6 are completely conceived and proven by the authors.  The numerics for this work were conceived by the authors, implemented by AI and verified by the authors~\cite{githublink}.



\let\oldaddcontentsline\addcontentsline
\renewcommand{\addcontentsline}[3]{}

\bibliography{prl}
\appendix
\section{Proof sketch of Lemma 3}\label{appsec:lemma3proof}
With some simple algebraic manipulation, as elaborated in the supplemental material, we can write
\begin{align}
    \pDist{N} \leq \frac{1+\eta}{2\eta}\;p_{\mathrm{Bin}}\mleft(X_M = \frac{M+N}{2}\mright).
\end{align}
Here we have introduced the binomially distributed random variable $X_M = \sum_{i=1}^M x_i$, where the $x_i$'s are IID random variables taking values $1$ with probability $(1+\eta)/2$ and $0$ with probability $(1-\eta)/2$. Therefore,
\begin{align}
    &\mathbb{P}_{\pDist{N}}\mleft(\abs{\frac{N}{M} - \eta}>\delta_1\mright) =\mathbb{P}_{\pDist{N}}\mleft(\abs{N - \eta M}>\delta_1 M\mright)\nonumber\\
    \leq\; &\frac{1+\eta}{2\eta}\; \mathbb{P}_{\mathrm{Bin}} \mleft(\abs{X_M - \frac{1+\eta}{2}M}> \frac{M\delta_1}{2}\mright)\nonumber\\
    \leq \; & \frac{1+\eta}{\eta}\;\exp\mleft(-\frac{M\delta_1^2}{2}\mright), 
\end{align}
and the last line comes from Hoeffding's inequality.\qed

\let\addcontentsline\oldaddcontentsline




\clearpage
\onecolumngrid
\begingroup

\let\theorem\supptheorem
\let\endtheorem\endsupptheorem

\let\lemma\supplemma
\let\endlemma\endsupplemma

\let\proposition\suppproposition
\let\endproposition\endsuppproposition

\let\corollary\suppcorollary
\let\endcorollary\endsuppcorollary

\let\conjecture\suppconjecture
\let\endconjecture\endsuppconjecture

\let\definition\suppdefinition
\let\enddefinition\endsuppdefinition

\let\remark\suppremark
\let\endremark\endsuppremark

\setcounter{secnumdepth}{3}
\renewcommand{\thesection}{\Roman{section}}
\setcounter{section}{0}
\setcounter{supptheorem}{0}

\renewcommand{\theHsection}{supp.\arabic{section}}

\numberwithin{equation}{section}
\setcounter{equation}{0}

\begin{center}
    {\large\bfseries Supplemental Material for}\\[4pt]
    {\large \bfseries ``A depolarizing choir sings in Gaussian harmony''}\\[10pt]
{\normalsize
 Rabsan Galib Ahmed$^{1}$, Sujeet Bhalerao$^2$, Sungjai Lee$^{1}$, Felix Leditzky$^{2}$,\\
Debbie Leung$^{1,3}$, Luke Schaeffer$^{1}$, and Graeme Smith$^{1}$
}\\[7pt]

{\small
$^1$University of Waterloo, Waterloo, Ontario, Canada.
}\\[3pt]

{\small
$^2$University of Illinois Urbana-Champaign, Urbana, Illinois, USA
}\\[3pt]

{\small
$^3$Perimeter Institute for Theoretical Physics, Waterloo, Ontario, Canada.
}\\[3pt]
\end{center}

\tableofcontents 

\section{Background and notation}

\subsection{Representation theory}

We write $\Lambda(M,2)$ for the set of partitions of $M$ into at most $2$ parts. The general linear group $\GL(d)$ is the set of invertible linear operators on $\C^d.$ For $\lambda\in\Lambda(M,2)$ we put
\begin{equation}
N=\lambda_1-\lambda_2,\qquad \lambda_2=\frac{M-N}{2},
\end{equation}
so $\lambda\mapsto N$ is a bijection from $\Lambda(M,2)$ onto
$\bigl\{N\in\{0,\dots,M\}: N\equiv M \bmod 2 \bigr\}$.
For $\pi\in S_M$ we consider the permutation representation $\pi\mapsto W_\pi$, where the unitary permutation operator $W_\pi$ on $(\mathbb{C}^2)^{\otimes M}$ acts as
\begin{align} 
    W_\pi\ket{x_1}\otimes\cdots\otimes\ket{x_M}=\ket{x_{\pi^{-1}(1)}}\otimes\cdots\otimes\ket{x_{\pi^{-1}(M)}}.
    \label{eq:symmetric-group-rep}
\end{align}
We also consider the representation of $\GL(2)$ on $(\mathbb{C}^2)^{\otimes M}$ as
\begin{align} 
    g&\mapsto g^{\otimes M}.
    \label{eq:general-linear-group-rep}
\end{align}

Since every $A^{\otimes M}$ with $A\in\mathcal{L}(\mathbb{C}^2)$ commutes with every $W_\pi$, the two representations \eqref{eq:symmetric-group-rep} and \eqref{eq:general-linear-group-rep} commute with each other.
Moreover, they span each other's commutant in $\cL\left((\mathbb{C}^2)^{\otimes M}\right)$, which is known as \emph{Schur-Weyl duality} \cite{goodman2009symmetry,etingof2011introduction,fulton2013representation}.
This duality also gives a decomposition of the common representation space $(\mathbb{C}^2)^{\otimes M}$ as
\begin{align}
\begin{aligned}
 (\mathbb{C}^2)^{\otimes M} &\cong
 \bigoplus_{\lambda \in \Lambda(M, 2)} V_{\lambda}^2\otimes S_\lambda\\
 U^{\otimes M}W_\pi &\cong\bigoplus_{\lambda\in\Lambda(M,2)} q_\lambda(U)\otimes p_\lambda(\pi),
 \end{aligned}
 \label{eq:schur-weyl-decomposition}
\end{align}
where the $(V_\lambda^2,q_\lambda)$ are irreducible representations (irreps) of $\GL(2)$, and $(S_\lambda,p_\lambda)$ are irreps of the symmetric group $S_M$. 
In the following discussion as well as in the main text, we will often use the following `spin notation' commonly used in the Physics literature alongside the partition notation introduced above: For a partition $\lambda=(\lambda_1,\lambda_2)\in \Lambda(M,2)$, we set $N=\lambda_1-\lambda_2$ so that $\lambda_1=(M+N)/2$ and $\lambda_2=(M-N)/2$, and we write $\glIrrep \equiv V_{(\lambda_1,\lambda_2)}^2$ and $\snIrrep \equiv S_{(\lambda_1,\lambda_2)}$.

For $\lambda\in \Lambda(M,2)$ we denote by $\Pi_\lambda$ the projector onto the isotypical component $V_\lambda^2\otimes S_\lambda$ in \eqref{eq:schur-weyl-decomposition}.
With $\lambda=(\lambda_1,\lambda_2)$, the dimension of $S_\lambda$ is equal to 
\begin{equation}
 f_\lambda \coloneqq \dim S_\lambda
 =\binom M{\lambda_2}-\binom M{\lambda_2-1} = \binom{M}{\frac{M-N}{2}} - \binom{M}{\frac{M-N}{2}-1},
 \label{eq:S-lambda-dimension}
\end{equation}
where we set $\binom M{-1}=0$. 
We also use the notation $\snIrrepDim\equiv f_\lambda$.

The symmetric subspace
\begin{align} 
\Sym^n(\mathbb{C}^2)=\{\ket{\psi}\in(\mathbb{C}^2)^{\otimes n}: W_\pi\ket{\psi}=\ket{\psi}\ \forall\pi\in S_n\}
\end{align}
has dimension $n+1$ and an orthonormal basis consisting of Dicke states is given by
\begin{equation}
\dicke{n}{k}=\binom{n}{k}^{-1/2}\smashoperator{\sum_{x\in\{0,1\}^n,\ |x|=k}}\ket{x},\qquad k=0,\dots n,
\label{eq:dicke}
\end{equation}
where $|x|$ is the Hamming weight of the string $x$. We call $k$ the excitation number of $\dicke{n}{k}$.
We write $P_n$ for the orthogonal projector onto $\Sym^n(\mathbb{C}^2)$, and note that
$P_nW_\pi=W_\pi P_n=P_n$ for all $\pi\in S_n$.
Since $A^{\otimes n}$ commutes with every $W_\pi$, it leaves $\Sym^n(\mathbb{C}^2)$ invariant,
and we write
\begin{equation}
S_n(A)\coloneqq A^{\otimes n}\big|_{\Sym^n(\mathbb{C}^2)},\qquad S_0(A)\coloneqq 1 .
\end{equation}
Then $S_n(AB)=S_n(A)S_n(B)$, and $A\mapsto S_n(A)$ restricted to $GL(2)$ is a representation. An explicit expression for the matrix entries of $S_n(A)$ is given in \cite[Appendix B]{bhalerao2025improvingquantumcommunicationrates}.
For $A=\mathrm{diag}(a,b)$ one has $S_n(A)\dicke{n}{k}=a^{n-k}b^k\dicke{n}{k}$, hence
$\tr S_n(A)=\sum_{k=0}^n a^{n-k}b^k$.
Finally, for $A\in\GL(2)$ (or $A\in\cL(\mathbb{C}^2)$ by continuity), the $\GL(2)$-irrep $q_\lambda$ for $\lambda=(\lambda_1,\lambda_2)$ can be expressed via the formula
\begin{equation}
 q_\lambda(A)=(\det A)^{(M-N)/2}S_N(A),
 \label{eq:determinant-twist}
\end{equation}
where as before we have $N=\lambda_1-\lambda_2$ and $(M-N)/2 = \lambda_2$.

\subsection{Quantum capacity}

A quantum channel $\mathcal{N}\colon A\to B$ is a completely
positive trace-preserving map from $\mathcal{L}(\mathcal{H}_A)$ to $\mathcal{L}(\mathcal{H}_B)$,
where $\mathcal{L}(\mathcal{H})$ denotes the set of linear operators on $\mathcal{H}$.
Every channel has a Stinespring isometry $V\colon\mathcal{H}_A\to\mathcal{H}_B\otimes\mathcal{H}_E$
with $\mathcal{N}(\rho)=\tr_E V\rho V^\dagger$, and the complementary channel is
$\mathcal{N}^c(\rho)=\tr_B V\rho V^\dagger$. 
If $\lbrace K_i\rbrace_i$ is a Kraus decomposition for $\cN$ with $\cN(X)=\sum_i K_i X K_i^\dagger$, then $V$ can be chosen as $V = \sum_i K_i\otimes |i\rangle_E$, where $\lbrace |i\rangle_E\rbrace_i$ is an orthonormal basis for the environment $E$.
The coherent information of a state $\rho$ through $\mathcal{N}$ is
\begin{equation}
I_c(\rho,\mathcal{N})\coloneqq S(\mathcal{N}(\rho))-S(\mathcal{N}^c(\rho)),
\label{eq:Ic}
\end{equation}
where $S(\rho)=-\tr\rho\log\rho$ is the von Neumann entropy and all logarithms are taken to base 2.
Equivalently, $I_c(\rho,\mathcal{N})=S(B)-S(RB)$ evaluated on
$(\mathrm{id}_R\otimes\mathcal{N})(\psi_{RA})$ for any purification $\psi_{RA}$ of $\rho$.
By the Araki--Lieb inequality, $|I_c(\rho,\mathcal{N})|\le S(\rho)\le\log\operatorname{rank}\rho$.
The quantum capacity is given by~\cite{lloyd1997capacity,shor2002quantum,devetak2005private}
\begin{equation}
Q(\mathcal{N})=\lim_{n\to\infty}\frac1n I_c(\mathcal{N}^{\otimes n})
=\sup_{n}\frac1n I_c(\mathcal{N}^{\otimes n}),
\label{eq:Q}
\end{equation}
where $I_c(\mathcal{N})\coloneqq \max_\rho I_c(\rho,\mathcal{N})$.
The regularization in Eq.~\eqref{eq:Q} cannot be removed in general because $I_c$ is
superadditive~\cite{shor1996quantumerrorcorrectingcodesneed,divincenzo1998capacity,smith2006degenerate}: there are channels and
states with $I_c(\rho_n,\mathcal{N}^{\otimes n})>nI_c(\mathcal{N})$.
Every state $\rho_n$ on $n$ channel inputs certifies the lower bound
$Q(\mathcal{N})\ge\frac1n I_c(\rho_n,\mathcal{N}^{\otimes n})$.
In the other direction, if $\mathcal{N}$ is antidegradable, meaning
$\mathcal{N}=\mathcal{A}\circ\mathcal{N}^c$ for some channel $\mathcal{A}$, then
$Q(\mathcal{N})=0$ by a no-cloning argument.

We parametrize the qubit depolarizing channel by its
transmissivity $\eta\in[0,1]$,
\begin{equation}
\cD_\eta(X)=\eta X+(1-\eta)\tr(X)\frac{\one_2}{2}=\eta X+b\,\cT(X),
\label{eq:Deta}
\end{equation}
where $b\coloneqq \frac{1-\eta}{2}$ and $\cT(X)\coloneqq \tr(X)\one_2$.
Equivalently, $\cD_\eta$ applies each of $X,Y,Z$ with probability $p/3$, where
$p=\frac34(1-\eta)$. For convenience, define \begin{equation}
 a=\frac{1+\eta}{2},\qquad b=\frac{1-\eta}{2},\qquad
 r=\frac ba=\frac{1-\eta}{1+\eta}.
 \label{eq:2}
\end{equation}

The depolarizing channel is unitarily covariant, that is, $$\cD_\eta(UXU^\dagger)=U \cD_\eta(X)U^\dagger$$ for all
$U\in U(2).$
The maximally mixed input gives the achievable rate~\cite{bennett1996mixed}
$$I_c(\one_2/2, \cD_\eta)=1-h(p)-p\log 3,$$ with $h(p)=-p\log p - (1-p)\log(1-p)$ the binary entropy.
It vanishes at $p\approx0.18929$ (or $\eta\approx0.7476$) which is known as the \emph{hashing bound}.
Superadditivity pushes the threshold $p_{\rm th}$ beyond this value.
Concatenated degenerate codes reach $p\approx0.1909$~\cite{smith2006degenerate} and
$p\approx0.1913$~\cite{fern2008lower}.
Recently, an optimization over rank-two states in the symmetric subspace of $n$ qubits yielded the improved lower bound
$p\approx0.1940$ (or $\eta\approx0.7414$)~\cite{agarwal2026enhanced}.

\subsection{Gaussian channels}
An $N$-mode bosonic system is defined by $N$ pairs $\mathbf{R}\coloneqq (x_1,p_1,\dots x_N,p_N)^T$ of \textit{canonical observables} satisfying the canonical commutation relation (CCR) $[R_j,R_k] = i\Omega_{jk}$, where $$\Omega \coloneqq  \begin{pmatrix}
    0 & 1\\
    -1 & 0
\end{pmatrix}^{\oplus N}.$$
Given a quantum state $\rho$ of an $N$-mode bosonic system, the first moment and the quantum covariance matrix (QCM) are defined by $d_j = \tr(\rho R_j)$ and $V_{jk} = \tr[(R_jR_k + R_kR_j)\rho] - 2d_j d_k$, respectively. The uncertainty principle implies that every QCM must satisfy $V\geq i\Omega$. A particularly important class of bosonic states called \textit{Gaussian states} is defined to be the ground and thermal states of Hamiltonians that are quadratic in $R_j$'s with eigenvalues bounded from below. Moreover, these states are entirely determined by their first moment and QCM~\cite{Simon1994,Serafini2017}.  

Quantum operations preserving Gaussianity are called Gaussian operations. Operationally, they can be realized using ancillary Gaussian states, Gaussian unitaries generated by quadratic Hamiltonians, and homodyne measurements. Important examples of Gaussian unitaries include phase shifts, squeezing transformations, beam splitters, and two-mode squeezers~\cite{Weedbrook2012}. Linear Bosonic Gaussian channels correspond to Gaussian operations where measurement outcomes are discarded~\cite{Holevo2001}.

 Among several examples of linear Gaussian channels, the two most relevant to this work are the single mode pure-loss attenuation channel, $\Loss_T$ and the quantum-limited amplification channel, $\Gain_{G}$. Mathematically, a pure-loss attenuation channel is realized by mixing the input mode with some environmental vacuum mode on a beam splitter of transmissivity $0\leq T \leq 1$, followed by tracing out the environment~\cite{eisert2005}. Physically, it arises when a fraction $1-T$ of a signal is absorbed. Its action on Gaussian states is $V\mapsto T \;V+(1-T)\mathbb{I}$ and $d\mapsto \sqrt{T}\;d$. 
 
 A quantum-limited amplification channel is mathematically realized by jointly acting on the input mode and environmental vacuum mode with a two-mode squeezer with gain $G\ge 1$, followed by tracing out the environment~\cite{eisert2005}. Its action on the QCM and the first moments is given by $V\mapsto G \;V+(G-1)\mathbb{I}$ and $d\mapsto \sqrt{G}\;d$. Physically, the signal is amplified by a factor $G$ with some added noise. 

 Fock states $\{\ket{k}: k\in \{0\}\cup\mathbb{N}\}$ span the Hilbert space of a single bosonic mode. The Kraus operators of $\Loss_T$ and $\Gain_G$ are conveniently expressed in the Fock basis
\begin{align}
    K^u_{T} &= \sum_{k=0}^\infty \sqrt{{k\choose u}T^{(k-u)}(1-T)^{u}}\; \ketbra{k-u}{k};\\
    A^t_{G} &= \sum_{k=0}^\infty \sqrt{{k+t \choose t}\frac{(G-1)^t}{G^{k+t+1}}}\; \ketbra{k+t}{k}
\end{align}
for $u, t = 0, \dots, \infty$.
Physically, the indices $u$ and $t$ respectively denote the number of photons lost to the environment and added by the environment during the action of the two channels. 

The composition of two Gaussian channels is also a Gaussian channel. For example, the action of $\Gaussian_{G,T} = \Gain_G\circ \Loss_T$ on the first two moments is given by $d\mapsto \sqrt{GT} d$ and $V\mapsto G\;[T\;V+ (1-T) \;\mathbb{I}] + (G-1)\;\mathbb{I}$. We distinguish three cases: (i) $GT <1$, (ii) $GT=1$, and (iii) $GT>1$, with the following effective bosonic channels: (i) a thermal attenuator  with transmissivity $GT$ and mean environmental photon number $(G-1)/(1-GT)$, (ii) an additive Gaussian noise with variance $(G-1)$, and (iii) a thermal amplifier with gain $GT$ and mean environmental photon number $G(1-T)/(GT-1)$. A comprehensive review of these channels can be found in~\cite{Ivan_Sabapathy_Simon}. The additive Gaussian noise channel of variance $\nu$ is a mixed unitary channel, given by
\begin{align}
    \mathcal{N}_{\nu} (\rho) = \frac{1}{2\pi \nu} \int_{\mathbb{C}}d^2\alpha \; e^{-\frac{\abs{\alpha}^2}{2\nu}} D(\alpha) \rho D^\dagger(\alpha),
\end{align}
 where $D(\alpha)$ is the unitary displacement operator in phase space.


\section{Depolarizing noise on the symmetric subspace}

In this work we focus on input states $\rho$ that are fully supported on the symmetric subspace $\Sym^M(\mathbb{C}^2)$, that is, $\rho = P_M\rho P_M$, with $P_M$ the projector onto $\Sym^M(\mathbb{C}^2)$.
In particular, such a $\rho$ is permutation-invariant and the action of an IID channel preserves the permutation invariance. Hence $\mathcal{D}_\eta^{\otimes M}(\rho)$ is also permutation-invariant, and by Schur's Lemma decomposes into irrep blocks $\Pi_{\lambda} \cD_\eta^{\otimes M} (\cdot) \Pi_{\lambda}$ on $V_{\lambda}^2$. We now show that the probability of mapping onto these blocks is independent of the input $\rho$.
\begin{lemma}\label{lem:sector-probability-input-ind}
Let $\Pi_\lambda$ be a projection onto $V_\lambda^2\otimes S_\lambda$ for $\lambda = (\lambda_1,\lambda_2)\in\Lambda(M,2)$, and set $N=\lambda_1-\lambda_2$. For every state $\rho$ on $\Sym^M(\mathbb{C}^2)$, the probability
\begin{align}
\pDist{N} =\tr\bigl[\Pi_\lambda\mathcal{D}_\eta^{\otimes M}(\rho)\bigr]
\label{eq:pN}
\end{align}
is independent of $\rho$.\footnote{Throughout the manuscript we are suppressing the dependence on $\eta$ of scalar quantities like $p_N$ to increase readability.
Furthermore, we often also suppress dependence on $M$, e.g., for $p_N$ defined in \eqref{eq:pN} or $\pi_s$ defined in \eqref{eq:explicit-mixing-weights}.}
\end{lemma}
\begin{proof}
Let $P_M$ be the projector onto the symmetric subspace
$\Sym^M(\mathbb C^2) = V^2_{(M)} = \glIrrep[M].$
Since $\rho$ is supported on $V^2_{(M)}$, we have
$\rho=P_M\rho P_M.$
Define $$X\coloneqq P_M(\mathcal D_\eta^{\otimes M})^*(\Pi_\lambda)P_M.$$
By the definition of the adjoint channel,
$$\begin{aligned}
\pDist{N}
&=\tr\!\left[\Pi_\lambda\mathcal D_\eta^{\otimes M}(\rho)\right] \\
&=\tr\!\left[(\mathcal D_\eta^{\otimes M})^*(\Pi_\lambda)\rho\right] \\
&=\tr[X\rho].
\end{aligned}$$

We show that $X$ is a scalar multiple of $P_M$. The depolarizing channel is unitarily covariant: for every $U\in U(2)$,
$$\mathcal D_\eta(U\cdot U^\dagger)=U\mathcal D_\eta(\cdot)U^\dagger.$$
Hence the adjoint of its $M$-fold tensor power satisfies
$$(\mathcal D_\eta^{\otimes M})^*\!\left(U^{\otimes M}\cdot(U^\dagger)^{\otimes M}\right)
=
U^{\otimes M}(\mathcal D_\eta^{\otimes M})^*(\cdot)(U^\dagger)^{\otimes M}.$$

The isotypical projectors $\Pi_\lambda$ are invariant under the collective $U(2)$-action as well, 
$$U^{\otimes M}\Pi_\lambda(U^\dagger)^{\otimes M}=\Pi_\lambda.$$
Applying the covariance of the adjoint then gives
$$U^{\otimes M}(\mathcal D_\eta^{\otimes M})^*(\Pi_\lambda)(U^\dagger)^{\otimes M}
=
(\mathcal D_\eta^{\otimes M})^*(\Pi_\lambda),$$
and similarly we also have $U^{\otimes M}X(U^\dagger)^{\otimes M}=X.$
Since the $U(2)$-representation $V^2_{(M)}=\Sym^M(\mathbb C^2)$ is irreducible, Schur's lemma implies that there is a scalar $c$ such that $X=cP_M$. 
Taking traces and using $P_M\rho P_M=\rho$, we get that $\pDist{N}=\tr[X\rho]=c\,\tr[P_M\rho]=c,$ and thus $\pDist{N}$ is independent of the choice of $\rho$.

\end{proof}

Write $S_N(A)$ for the action of $A$ on $\glIrrep = V_{\lambda}^2$ where $N = \lambda_1 - \lambda_2$. For $A=\diag(a,b)$, define $q_N=\tr S_N(A)$. Then
\begin{equation}
 q_N=\sum_{k=0}^N a^{N-k}b^k
 =\frac{a^{N+1}-b^{N+1}}{\eta},
 \label{eq:8}
\end{equation}
with the rightmost expression extended by continuity at $\eta=0$ (see, e.g., \cite[App.~B]{bhalerao2025improvingquantumcommunicationrates}). We now compute the input-independent probability by choosing a coherent input whose noisy output is diagonal.

\begin{lemma}\label{lem: p_MN-prob}
Let $\lambda=(\lambda_1,\lambda_2)\in\Lambda(M,2)$ and set $N=\lambda_1-\lambda_2$.
Then $\pDist{N}$ is equal to
\begin{equation}
 \pDist{N}=f_\lambda(ab)^{\lambda_2}q_N,
 \label{eq:9}
\end{equation}
where $f_\lambda \equiv \snIrrepDim$ is defined in \eqref{eq:S-lambda-dimension}.
Equivalently,
\begin{equation}
 \pDist{N}= \begin{cases}\dfrac{f_\lambda}{\eta}
 \left(a^{\lambda_1+1}b^{\lambda_2}-a^{\lambda_2}b^{\lambda_1+1}\right) & \text{if $\eta>0$,}\\[1em]
 \dfrac{f_\lambda(N+1)}{2^M} & \text{if $\eta=0$.}
 \end{cases}
 \label{eq:10}
\end{equation}
\end{lemma}
\begin{proof}

By Lemma~\ref{lem:sector-probability-input-ind}, the probability $\pDist{N}$ of any $\lambda$-block is independent of the input state in $\mathcal{Q}_{M,M}$. We may therefore evaluate it on the state $\ketbra{0}{0}^{\otimes M}$. Since $\mathcal D_\eta(\ketbra{0}{0})=A=\diag(a,b)$, its output is $A^{\otimes M}$.

Under Schur-Weyl duality, $A^{\otimes M}$ acts on each $\lambda$-block $\glIrrep \otimes \snIrrep$ as $q_\lambda(A)\otimes \one_{S_\lambda}$. Hence
$$\pDist{N} =\tr[\Pi_\lambda A^{\otimes M}]
=f_\lambda\,\tr[q_\lambda(A)],$$ where $f_\lambda=\dim S_\lambda$ is defined in \eqref{eq:S-lambda-dimension}.

By~\eqref{eq:determinant-twist}, $q_\lambda(A)=(\det A)^{\lambda_2}S_N(A)$. Since $A=\diag(a,b)$, we have $\det A=ab$, and $S_N(A)$ has eigenvalues $a^{N-k}b^k$ for $k=0,\ldots,N$ (see \cite[App.~B]{bhalerao2025improvingquantumcommunicationrates}). Therefore
$$\pDist{N}
=f_\lambda(ab)^{\lambda_2}\sum_{k=0}^N a^{N-k}b^k
=f_\lambda(ab)^{\lambda_2}q_N.$$

For $\eta>0$, since $a-b=\eta$,
$q_N=(a^{N+1}-b^{N+1})/\eta$. Using $N=\lambda_1-\lambda_2$ gives
$$\pDist{N}
=\frac{f_\lambda}{\eta}
\left(
a^{\lambda_1+1}b^{\lambda_2}
-a^{\lambda_2}b^{\lambda_1+1}
\right).$$

At $\eta=0$, we have $a=b=1/2$, so $q_N=(N+1)2^{-N}$. Since $M=N+2\lambda_2$,
$$\pDist[M][0]{N}
=f_\lambda 2^{-2\lambda_2}(N+1)2^{-N}
=\frac{f_\lambda(N+1)}{2^M},$$
which concludes the proof.
\end{proof}

The following theorem gives the Schur-Weyl block decomposition of
the channel output for an input supported on the symmetric subspace.
Since this is used in the main text, we phrase it in spin notation, using the abbreviation $\one_{M,N}\equiv \one_{S_\lambda}$ and $\snIrrepDim = f_\lambda$.

\begin{theorem}\label{thm:decomp}
    For any operator $X\in \cL(\mathrm{Sym}^M(\mathbb{C}^2))$, we have
    \begin{align}
        \mathcal{D}_\eta^{\otimes M}(X) = \bigoplus_{N}\; \pDist{N}\;\PhiChannel(X) \otimes \frac{\one_{M,N}}{\snIrrepDim},
        \label{eq:normalized-block-decomposition}
    \end{align}
    where $N=0,2,\dots,M$ for even $M$ and $N=1,3,\dots,M$ for odd $M$,
    the $\PhiChannel \colon {\cal L}(\mathrm{Sym}^M(\mathbb{C}^2))\to {\cal L}(\mathrm{Sym}^N(\mathbb{C}^2))$ are quantum channels, $\snIrrepDim$ is defined in \eqref{eq:S-lambda-dimension}, and $\pDist{N}$ is a fixed probability distribution over $N$ for each $M$, with
    \begin{align}
        \pDist{N} = \snIrrepDim \left(\frac{1-\eta^2}{4}\right)^{\frac{M-N}{2}} \frac{(1+\eta)^{N+1}-(1-\eta)^{N+1}}{2^{N+1}\eta}.
        \label{eq:block-probabilities}
    \end{align}
\end{theorem}

\begin{proof}
    The input state $\rho$ is supported on $\Sym^M(\mathbb{C}^2)$ and thus permutation-invariant, and the IID channel $ \mathcal{D}_p^{\otimes M}$ preserves this permutation-invariance of the input state $\rho.$ 
    Thus, the output $\mathcal{D}_p^{\otimes M}(\rho)$ is again a permutation-invariant state and by Schur-Weyl duality can be written as $$\mathcal D_{\eta}^{\otimes M}(X) =\bigoplus_N\unnormalPhiChannel(X)\otimes\frac{\one_{M,N}}{\snIrrepDim}, $$ where each $\unnormalPhiChannel$ is completely positive. By \Cref{lem:sector-probability-input-ind} and linearity, ${\tr \unnormalPhiChannel(X) = \pDist{N} \tr X}.$ Thus $\PhiChannel = \unnormalPhiChannel / \pDist{N} $ is a quantum channel whenever $\pDist{N} > 0.$

    Finally, \Cref{lem: p_MN-prob} gives the stated probabilities and their values at $\eta = 0$ after substituting $a=(1+\eta)/2$, $b=(1-\eta)/2$,
and $\lambda_2=(M-N)/2$.
\end{proof}

\subsection{Decomposition of the channel on irreps}
We will describe the channels $\PhiChannel$ as a concatenation of two simple channels.
The first one retains $s$ of the symmetric input qubits and discards the others.
\begin{definition}
For $0\le s\le M$, let
\begin{equation}
 \traceChannel(X) =\tr_{M-s}(X),
 \label{eq:17}
\end{equation}
viewed as a channel from $\mathcal{L}(\glIrrep[M])$ to $\mathcal{L}(\glIrrep[s][s])$.
\end{definition}

The second operation is the universal symmetric cloning channel of~\cite{Werner1998}. It enlarges the symmetric system by adjoining identities and projecting back onto the symmetric space.

\begin{definition}
For $0\le s\le N$, define
\begin{equation}
 \cloneChannel(X)=\frac{s+1}{N+1}
 P_N\bigl(X\otimes \one_2^{\otimes(N-s)}\bigr)P_N,
 \label{eq:18}
\end{equation}
viewed as a channel from $\mathcal{L}(\glIrrep[s][s])$ to $\mathcal{L}(\glIrrep[N][N])$.
\end{definition}

The following result is standard (see, e.g., \cite{chiribella,harrow2013church}).
\begin{proposition}\label{prop:finite-loss-cloner-channels}
The maps $\traceChannel$ and $\cloneChannel$ are
completely positive and trace preserving.
\end{proposition}

We can now compute the action of the channels $\PhiChannel$ on tensor powers of rank-one operators. Such
operators suffice because their linear span is the entire operator space, $\mathcal{L}(\mathrm{Sym}^M(\mathbb{C}^2))$.

\begin{proposition}\label{prop:coherent-kernel-finite}
Let $\ket{u}, \ket{v} \in\mathbb C^2$ be unit vectors, set
$g=\langle v|u\rangle$, and put
$A_{u,v}=\mathcal D_\eta(\ketbra{u}{v})$.
The operators $\ketbra{u}{v}^{\otimes M}$ span
$\mathcal{L}(\glIrrep[M])$. For every block with positive probability,
\begin{equation}\label{eq:finite-gl2-kernel}
 \PhiChannel
   \mleft(\ketbra{u}{v}^{\otimes M} \mright)
 =\frac{g^{M-N}}{q_N}S_N(A_{u,v}).
\end{equation}
\end{proposition}
\begin{proof}
    Note $\cD_\eta (\ketbra{u}{v}) =: A_{u,v} = bg \id_2 + \eta \ketbra{u}{v}$ with $b = \frac{1 - \eta}{2}$.  Because ${A_{u,v}^{\otimes M} = \cD_\eta^{\otimes M} (\ketbra{u}{v}^{\otimes M} ) }$, the output of $\PhiChannel(\ketbra{u}{v}^{\otimes M})$ is given by $q_{\lambda}(A_{u,v})$, the $\GL(2)$-irrep of $A_{u,v}$ on $\glIrrep$.
    
    The determinant for matrices of the form $x\id_2 + y R$ for $\det R = 0$ is given by $x^2 + xy \tr R$, giving
    \begin{equation*}
        \det A_{u,v} = b^2g^2 + b\eta g^2= abg^2
    \end{equation*}
    for $a = \frac{1 + \eta}{2}$. Consequently, according to \eqref{eq:determinant-twist}, 
    \begin{equation*}
        q_\lambda(A_{u,v}) = (ab)^{\lambda_2} g^{2\lambda_2} S_N(A_{u,v}).
    \end{equation*}
    
For $0\leq\eta<1$, comparison with \eqref{eq:9} and the normalized block decomposition \eqref{eq:normalized-block-decomposition}
now gives
\begin{align}
 \PhiChannel (\ketbra{u}{v}^{\otimes M})
 &=\frac{f_\lambda q_\lambda(A_{u,v})}{\pDist{N}}=\frac{g^{2\lambda_2}}{q_N}S_N(A_{u,v})
 =\frac{g^{M-N}}{q_N}S_N(A_{u,v}).
\end{align}
\end{proof}

Define now for $X\in\mathcal{L}(\glIrrep)$ the map
\begin{equation}\label{eq:top-sector-thermalizer}
 \Theta_{N,\eta}(X)
    =\frac{1}{q_N}P_N\mathcal D_\eta^{\otimes N}(X)P_N.
\end{equation}
Since $q_N>0$, this map is completely positive. Applying the sector
probability formula with input size $N$ and partition $(N,0)$ shows
that the operator $P_N\mathcal D_\eta^{\otimes N}(X)P_N$ in \eqref{eq:top-sector-thermalizer} has
trace $q_N \tr X$. Thus $\Theta_{N,\eta}$ is also
trace preserving. The next identity reduces every block map $\Phi^\eta_{M\to N}$ to a composition of a partial trace with the map in \eqref{eq:top-sector-thermalizer}.

\begin{proposition}\label{prop:exact-first-factorization}
For $0\le\eta\le1$,
\begin{equation}\label{eq:first-factorization}
 \PhiChannel=\Theta_{N,\eta}\circ \traceChannel[N].
\end{equation}
\end{proposition}

\begin{proof}
Again, it suffices to compare
the two maps on $\ketbra{u}{v}^{\otimes M}$, where $g= \langle{v|u}\rangle$ and $A_{u,v} = \cD_\eta(\ketbra{u}{v})$. The partial trace and the definition of
$\Theta_{N,\eta}$ give
\begin{align*}
 (\Theta_{N,\eta}\circ\traceChannel[N])(\ketbra{u}{v}^{\otimes M})
 &=g^{M-N}\Theta_{N,\eta}(\ketbra{u}{v}^{\otimes N})\\
 &=\frac{g^{M-N}}{q_N}
       P_N A_{u,v}^{\otimes N}P_N\\
 &=\frac{g^{M-N}}{q_N}S_N(A_{u,v}) \\
 &= \PhiChannel(\ketbra{u}{v}^{\otimes M})
\end{align*}
where the final equality is given by \eqref{eq:finite-gl2-kernel}.
\end{proof}

The next theorem decomposes each normalized block channel into
a mixture of compositions of $\cloneChannel$ and
$\traceChannel$. 
The decomposition is based on the observation that once restricted to the symmetric subspace, the depolarizing channel's output can be manipulated similarly to binomials.

\begin{theorem}
\label{thm:block-channel-decomposition}
The conditional channels satisfy
    \begin{align}
    \PhiChannel
    =
    \sum_{s=0}^{N}
    \piDist{s}\,
    \cloneChannel
    \circ
    \traceChannel,
    \end{align}
where the probability distribution over $s$, the number of retained qubits, is given by 
\begin{align}\label{eq:explicit-mixing-weights}
    \piDist{s} = {N+1 \choose s+1}\frac{(2\eta)^{s+1}(1-\eta)^{N-s}}{(1+\eta )^{N+1}- (1-\eta)^{N+1}},
\end{align}
and $L_{M\to s}\colon \mathrm{Sym}^M(\mathbb{C}^2)\to \mathrm{Sym}^s(\mathbb{C}^2)$ and $C_{s\to N} : \mathrm{Sym}^s(\mathbb{C}^2)\to \mathrm{Sym}^N(\mathbb{C}^2)$ are respectively the loss channel and the cloning channel.
\end{theorem}
\begin{proof}
 We will first
decompose $\Theta_{N,\eta}$ and then make use of  ~\Cref{prop:exact-first-factorization}. Write
$$\mathcal D_\eta=\eta\operatorname{id}+b \cT,$$ where
$b=(1-\eta)/2$ and $ \cT(X)=\tr(X)\id_2$. Write $P_N$ for the projector onto $\Sym^N(\C^2)$. Expanding $\mathcal D_\eta^{\otimes N}$ using this expression for $\cD_\eta$, we obtain 
\begin{align*}
 P_N\mathcal D_\eta^{\otimes N}(X)P_N
 &=\sum_{s=0}^N\binom Ns\eta^s b^{N-s}
   P_N\bigl(\traceChannel[s][N](X)\otimes
                I_2^{\otimes(N-s)}\bigr)P_N\\
 &=\sum_{s=0}^N\binom{N+1}{s+1}\eta^s b^{N-s}
   \cloneChannel \bigl(\traceChannel[s][N](X)\bigr).
\end{align*}
The second equality uses the normalization in~\eqref{eq:18} and
$\binom{N}{s} \frac{N+1}{s+1}=\binom{N+1}{s+1}$. Dividing by $q_N$ gives
\begin{align}
 \Theta_{N,\eta}
 =\sum_{s=0}^N
   \underbrace{\frac{\binom{N+1}{s+1}\eta^s b^{N-s}}{q_N}}
   \cloneChannel \circ \traceChannel[s][N].
   \label{eq:Theta-expansion}
\end{align}
For $\eta>0$, the identity
$q_N=(a^{N+1}-b^{N+1})/\eta$, with $a=(1+\eta)/2$ and $b=(1-\eta)/2$,
shows that the underbracketed coefficients in \eqref{eq:Theta-expansion} are exactly~\eqref{eq:explicit-mixing-weights}.
They are nonnegative, and the binomial theorem gives
$$
 \sum_{s=0}^N\binom{N+1}{s+1}\eta^s b^{N-s}
   =\frac{(b+\eta)^{N+1}-b^{N+1}}{\eta}
   =q_N.
$$
Thus the coefficients sum to one. Each summand is a channel by
Proposition~\ref{prop:finite-loss-cloner-channels}. Finally, pre-compose the decomposition of $\Theta_{N,\eta}$ with
$ \traceChannel[N]$. Since successive partial traces satisfy
$$
\traceChannel[s][N] \circ \traceChannel[N][M]
   = \traceChannel ,
$$
Proposition~\ref{prop:exact-first-factorization} proves the claimed identity.
\end{proof}

\section{Concentration lemmas}
In this section we prove two concentration estimates that concern the block probabilities $\pDist{N}$ defined in \eqref{eq:block-probabilities} and the mixing weights $\pi_s$ defined in \eqref{eq:explicit-mixing-weights}.
We start with the $\pDist{N}$.
\begin{lemma}
\label{lem:block-probability-concentration}
 For any $\delta_1 >0$ and every $M$,
    \begin{align}
        \mathbb{P}_{\pDist{N}}\left(\abs{\frac{N}{M} - \eta}>\delta_1\right) \leq \frac{1+\eta}{\eta}\;\exp\left(-\frac{M\delta_1^2}{2}\right).
    \end{align}
\end{lemma}
\begin{proof}
    We have \begin{align}
        \pDist{N} &= \left[{M \choose \frac{M-N}{2}} - {M \choose \frac{M-N}{2}-1}\right] \left(\frac{1-\eta^2}{4}\right)^{\frac{M-N}{2}} \frac{(1+\eta)^{N+1}-(1-\eta)^{N+1}}{2^{N+1}\eta}\nonumber\\
        &= {M \choose \frac{M-N}{2}}  \left(\frac{1+\eta}{2}\right)^{\frac{M-N}{2}} \left(\frac{1-\eta}{2}\right)^{\frac{M-N}{2}} \left(\frac{1+\eta}{2}\right)^{N} \frac{1+\eta}{2\eta} \left[1-\left(\frac{1-\eta}{1+\eta}\right)^{N+1}\right] \frac{2N+2}{M+N+2}\nonumber\\
        &= {M \choose \frac{M-N}{2}}  \left(\frac{1+\eta}{2}\right)^{\frac{M+N}{2}} \left(\frac{1-\eta}{2}\right)^{\frac{M-N}{2}}  \frac{1+\eta}{2\eta} \left[1-\left(\frac{1-\eta}{1+\eta}\right)^{N+1}\right] \frac{2N+2}{M+N+2}\nonumber\\
        &\leq {M \choose \frac{M-N}{2}}  \left(\frac{1+\eta}{2}\right)^{\frac{M+N}{2}} \left(\frac{1-\eta}{2}\right)^{\frac{M-N}{2}} \frac{1+\eta}{2\eta}\nonumber\\
        & = \frac{1+\eta}{2\eta}\;p_{\mathrm{Bin}}\left(X_M = \frac{M+N}{2}\right). 
    \end{align}
    In the last line we have introduced the random variable $X_M = \sum_{i=1}^M x_i$, where the $x_i$'s are IID random variables taking values $1$ with probability $(1+\eta)/2$ and $0$ with probability $(1-\eta)/2$. Therefore, for any $\delta_1>0$ and every $M$, 
    \begin{align}
        \mathbb{P}_{\pDist{N}}\left(\abs{\frac{N}{M} - \eta}>\delta_1\right) 
        &= \mathbb{P}_{\pDist{N}}\left(\abs{N - \eta M}>\delta_1 M\right)\nonumber\\
        &= \sum_{N \notin \bigl[(  \eta-\delta_1)M, (\eta+\delta_1)M \bigr]} \pDist{N} \nonumber\\
        &\leq \frac{1+\eta}{2\eta}\sum_{X_M \notin \bigl[\frac{(1+\eta-\delta_1)M}{2},\frac{(1+\eta+\delta_1)M}{2}\bigr]} p_{\mathrm{Bin}}(X_M)\nonumber\\
        &= \frac{1+\eta}{2\eta}\; \mathbb{P}_{\mathrm{Bin}}\left(\abs{X_M - \frac{1+\eta}{2}M}> \frac{M}{2}\delta_1\right)\nonumber\\
        &\leq \frac{1+\eta}{\eta} \exp\left(-\frac{2 M^2\delta_1^2}{4M}\right)\\
        &= \frac{1+\eta}{\eta} \exp\left(-\frac{ M\delta_1^2}{2}\right),
    \end{align}
    which concludes the proof.
\end{proof}

We now prove a concentration result for the $\piDist{s}$.
\begin{lemma}
\label{lem:mixing-weight-concentration}
   For any $1\geq \delta_2 >0$, and $c=e^{4\frac{1-\eta}{1+\eta}}$, we have that for every $N>\frac{1}{\delta_2}\frac{1-\eta}{1+\eta}$, 
    \begin{align}\label{eq:con.in.s}
\mathbb{P}_{\piDist{s}}\left(\abs{\frac{s}{N} - \frac{2\eta}{1+\eta}}>\delta_2\right) \leq \frac{1+\eta}{\eta}c\;\exp\left(-\frac{2N^2\delta_2^2}{(N+1)}\right).
    \end{align}
\end{lemma}
\begin{proof}
    We have
    \begin{align}
        \piDist{s} &= {N+1 \choose s+1} \frac{(2\eta)^{s+1}(1-\eta)^{N-s}}{(1+\eta)^{N+1}-(1-\eta)^{N+1}}\nonumber\\
        & = {N+1 \choose s+1} \left(\frac{2\eta}{1+\eta}\right)^{s+1} \left(\frac{1-\eta}{1+\eta}\right)^{(N+1)-(s+1)}\frac{1}{1-\left(\frac{1-\eta}{1+\eta}\right)^{N+1}}\nonumber\\
        & \leq {N+1 \choose s+1} \left(\frac{2\eta}{1+\eta}\right)^{s+1} \left(\frac{1-\eta}{1+\eta}\right)^{(N+1)-(s+1)} \frac{1}{1-\frac{1-\eta}{1+\eta}}\nonumber\\
        & = \frac{1+\eta}{2\eta} p_{\mathrm{Bin}}(Y_{N+1} = s+1).
    \end{align}
    In the last line we have introduced the random variable $Y_{N+1} = \sum_{i=1}^{N+1} y_i$, where the $y_i$'s are IID random variables taking values $1$ with probability $2\eta/(1+\eta)$ and $0$ with probability $(1-\eta)/(1+\eta)$. Therefore, for any $\delta_2>0$ and every $N$,
    \begin{align}
        \mathbb{P}_{\piDist{s}} \left(\abs{\frac{s}{N}-\frac{2\eta}{1+\eta}}> \delta_2\right) &= \mathbb{P}_{\piDist{s}} \left(\abs{s-\frac{2\eta}{1+\eta}N}> N\delta_2\right)\nonumber\\
        &= \sum_{s \notin \bigl[\left(\frac{2\eta}{1+\eta}-\delta_2\right)N,\left(\frac{2\eta}{1+\eta}+\delta_2\right)N\bigr] } \piDist{s}\nonumber.
    \end{align}
    The intervals for $s$ translate to the following intervals for $Y_{N+1}$. The lower bound is given by
    \begin{align}
        \left(\frac{2\eta}{1+\eta}-\delta_2\right)N + 1 &= \left(\frac{2\eta}{1+\eta}-\frac{\delta_2N}{N+1} + \frac{1}{N+1} -\frac{1}{N+1}\frac{2\eta}{1+\eta}\right)(N+1) \nonumber\\
        &= \left(\frac{2\eta}{1+\eta} - \frac{\delta_2N}{N+1} +\frac{1}{N+1}\frac{1-\eta}{1+\eta}\right)(N+1),
    \end{align}
    while the upper bound is given by
    \begin{align}
        \left(\frac{2\eta}{1+\eta}+\delta_2\right)N + 1 &= \left(\frac{2\eta}{1+\eta} + \frac{\delta_2N}{N+1} +\frac{1}{N+1}\frac{1-\eta}{1+\eta}\right)(N+1).
    \end{align}
    Now the total probability that $Y_{N+1}$ is outside the interval $$\left[ \left(\frac{2\eta}{1+\eta} - \frac{\delta_2N}{N+1} +\frac{1}{N+1}\frac{1-\eta}{1+\eta}\right)(N+1),\left(\frac{2\eta}{1+\eta} + \frac{\delta_2N}{N+1} +\frac{1}{N+1}\frac{1-\eta}{1+\eta}\right)(N+1) \right]$$
    is less than or equal to the total probability that $Y_{N+1}$ is outside the interval $$\left[ \left(\frac{2\eta}{1+\eta} - \frac{\delta_2N}{N+1} +\frac{1}{N+1}\frac{1-\eta}{1+\eta}\right)(N+1),\left(\frac{2\eta}{1+\eta} + \frac{\delta_2N}{N+1} -\frac{1}{N+1}\frac{1-\eta}{1+\eta}\right)(N+1) \right].$$
    Let us denote, $\delta'_N = \frac{\delta_2N}{N+1} -\frac{1}{N+1}\frac{1-\eta}{1+\eta}>0$. Then we have
    \begin{align}
        &\mathbb{P}_{\piDist{s} } \left(\abs{\frac{s}{N}-\frac{2\eta}{1+\eta}}> \delta_2\right) \nonumber\\
        &\leq \frac{1+\eta}{2\eta} \sum_{Y_{N+1}\notin \bigl[\left(\frac{2\eta}{1+\eta} -\delta'_N\right)(N+1),\left(\frac{2\eta}{1+\eta} +\delta'_N\right)(N+1)\bigr]} p_{\mathrm{Bin}}(Y_{N+1}) \nonumber\\
        &= \frac{1+\eta}{2\eta}\; \mathbb{P}_{\mathrm{Bin}}\left(\abs{Y_{N+1}- \frac{2\eta}{1+\eta}(N+1)}> (N+1)\delta'_N\right) \nonumber\\
        &\leq \frac{1+\eta}{\eta} \exp\left(-2(N+1)\delta_N'^2\right)\nonumber\\
        &= \frac{1+\eta}{\eta} \exp\left(-\frac{2N^2\delta^2_2}{N+1}\right)\exp\left(\frac{4N\delta_2}{N+1}\frac{1-\eta}{1+\eta}\right)\exp\left(-\frac{2}{N+1}\left(\frac{1-\eta}{1+\eta}\right)^2\right) \nonumber\\
        &\leq \frac{1+\eta}{\eta} \exp\left(-\frac{2N^2\delta^2_2}{N+1}\right)\exp\left(\frac{4N\delta_2}{N+1}\frac{1-\eta}{1+\eta}\right) \nonumber\\
        &\leq \frac{1+\eta}{\eta}e^{4\delta_2\frac{1-\eta}{1+\eta}} \exp\left(-\frac{2N^2\delta^2_2}{N+1}\right).
    \end{align}
    As we have $0<\delta_2\leq 1$ and $c = e^{4\frac{1-\eta}{1+\eta}} $, we finally obtain
    \begin{align}
        &\;\mathbb{P}_{\piDist{s} } \left(\abs{\frac{s}{N}-\frac{2\eta}{1+\eta}}> \delta_2\right)\leq \frac{1+\eta}{\eta}c\; \exp\left(-\frac{2N^2\delta^2_2}{N+1}\right),
    \end{align}
    proving the claim.
\end{proof}

\section{Gaussian limit}
\subsection{Gaussian limit of the random loss-cloning channels}

In this section, we show convergence of $\PhiChannel$ to the Gaussian channels $\Gain_{\frac{1 + \eta}{2\eta}} \circ \Loss_{\frac{2\eta}{1 + \eta} \frac{N}{M}}$. As we talk about linear maps between infinite dimensional vector spaces, we work with the following norm.

\begin{definition}[Diamond norm with a cutoff, $K$] The diamond norm of the restriction of a linear map $\Phi$ to the first $K$ excitations is denoted by $\norm{\Phi}_{\diamond,K}\coloneqq \norm{\Phi\vert_{\mathcal{F}_K}}_{\diamond}$, where ${ \mathcal{F}_K =\mathrm{Span}\{\ket{0},\dots, \ket{K}\} }.$ 
\end{definition}
As the input space of this restricted channel is finite-dimensional, this constitutes a valid norm~\cite{Paulsen_2003}.

\begin{theorem}\label{thm:diamond-dist-bound} For a fixed excitation cutoff, $K\leq M$, and a fixed ratio $q= N/M$ with $0<q<1$,
    \begin{align}
        \norm{\PhiChannel - \mathcal{A}_{\frac{1+\eta}{2\eta}}\circ \mathcal{L}_{\frac{2\eta}{1+\eta}\frac{N}{M}}}_{\diamond, K} \leq \varepsilon_{M,K,q},
    \end{align}
    where $\varepsilon_{M,K,q}\to 0$ as $M\to \infty$.
\end{theorem}
The remainder of this section is dedicated to proving this theorem.
First, using the decomposition in Theorem~\ref{thm:block-channel-decomposition}, we can write
\begin{align}
    \PhiChannel - \mathcal{A}_{\frac{1+\eta}{2\eta}}\circ \mathcal{L}_{\frac{2\eta}{1+\eta}\frac{N}{M}} = \sum_{s=0}^N \piDist{s} \left(\cloneChannel
    \circ
    \traceChannel- \mathcal{A}_{\frac{1+\eta}{2\eta}}\circ \mathcal{L}_{\frac{2\eta}{1+\eta}\frac{N}{M}}\right).
\end{align}
Therefore, using the triangle inequality,
\begin{align}\label{eq:tri-diamond}
    &\;\norm{\PhiChannel - \mathcal{A}_{\frac{1+\eta}{2\eta}}\circ \mathcal{L}_{\frac{2\eta}{1+\eta}\frac{N}{M}}}_{\diamond, K}
    \leq \sum_{s=0}^N \piDist{s} \norm{\cloneChannel
    \circ
    \traceChannel- \mathcal{A}_{\frac{1+\eta}{2\eta}}\circ \mathcal{L}_{\frac{2\eta}{1+\eta}\frac{N}{M}}}_{\diamond, K}. 
\end{align}

As we put a cutoff at the excitation $K$, we are in fact calculating the diamond distance between the channels that we obtain by restricting $\cloneChannel \circ \traceChannel$ and $\mathcal{A}_{G}\circ \mathcal{L}_T$ to the first $K$ excitation levels. Denoting these channels by $\cloneChannel \circ \traceChannel \vert_{K}$ and $\mathcal{A}_{G}\circ \mathcal{L}_T\vert_K$, we can write 
\begin{align} 
    \norm{ \big. \cloneChannel\circ \traceChannel - \mathcal{A}_{G}\circ \mathcal{L}_T}_{\diamond,K} = \norm{\big. \cloneChannel\circ \traceChannel\vert_{K} - \mathcal{A}_{G}\circ \mathcal{L}_T\vert_{K}}_{\diamond}.
\end{align}

Furthermore, for an input state $\rho$, supported within the excitation cutoff $K$, both $\mathcal{L}_T(\rho)$ and $\traceChannel(\rho)$ are supported within the excitation cutoff $K$. Hence, we can write $\cloneChannel \circ \traceChannel \vert_{K} = \cloneChannel \vert_{K}\circ L_{M\to s}\vert_{K}$ and $\mathcal{A}_{G}\circ \mathcal{L}_T\vert_{K} = \mathcal{A}_{G}\vert_{K}\circ \mathcal{L}_T\vert_{K}$. Using the diamond norm inequality for the composition of channels~\cite{john-watrous-book}, we have
\begin{align}
    \;\norm{ \cloneChannel\circ \traceChannel - \mathcal{A}_{G}\circ \mathcal{L}_T}_{\diamond,K} 
    =&\;\norm{\cloneChannel \vert_K\circ \traceChannel \vert_{K} - \mathcal{A}_{G}\vert_{K}\circ \mathcal{L}_T\vert_{K}}_{\diamond}\nonumber\\
    \leq &\; \norm{\cloneChannel\vert_K -  \mathcal{A}_{G}\vert_{K}}_{\diamond} + \norm{\traceChannel \vert_{K}-\mathcal{L}_T\vert_{K}}_{\diamond}\nonumber\\
    = &\; \KDiaNorm{\cloneChannel - \Gain_G} + \KDiaNorm{\traceChannel - \Loss_T}.\label{eq:decomp-into-loss-and-gain}
\end{align}
We bound each of these diamond distances separately in Lemmas \ref{lem:clone-channel-gain-channel-diamond-norm-bound} and \ref{lem:trace-channel-loss-channel-diamond-norm-bound}, respectively. The bounds can be understood simply in the following manner. 
First recall that the Kraus operators of these channels are given by
\begin{align}
    \cloneChannel &\sim \left\{ C^{t}_{s\to N} = \sum_{i = 0}^{s} \sqrt{ \frac{s + 1}{N + 1} \frac{\binom{s}{i} \binom{N - s}{t}}{\binom{N}{i + t} }} \ketbra{i + t}{i} \right\}_{t = 0, \dots, N - s} \\
    \Gain_G &\sim \left\{ A_{G}^t = \sum_{i= 0}^\infty \sqrt{\binom{i+t}{t} \frac{(G - 1)^t}{G^{i + t + 1}} } \ketbra{i+t}{i} \right\}_{t = 0, \dots, \infty} \\
    \traceChannel &\sim \left\{ L_{M \to s}^{u} = \sum_{i = u}^{M} \sqrt{\frac{ \binom{s}{i - u} \binom{M - s}{u} }{\binom{M}{i}}} \ketbra{i - u}{i} \right\}_{u = 0, \dots, M - s} \\
    \Loss_T &\sim \left\{ K_T^u = \sum_{i= u}^\infty \sqrt{\binom{i}{u} T^{i - u} (1-T)^u } \ketbra{i - u}{i} \right\}_{u = 0, \dots, \infty} \label{eq:loss-channel-kraus-operators}
\end{align}
The coefficients of the Kraus operators in the excitation/Fock basis for the loss channel and the Gaussian attenuation channel respectively resemble square-roots of a hypergeometric distribution and a binomial distribution. Similarly, those of the cloning channel and the Gaussian amplification channel respectively resemble square-roots of a negative hypergeometric distribution and a negative binomial distribution. Within each pair of the distributions, convergence occurs under fixed sample size and number of successes respectively~\cite{Kinney2009-oq}, as the population size grows. Furthermore as $N$ grows, within the typical interval for $s/N$ around $\frac{2\eta}{1+\eta}$, we see a decaying upper bound for each of the two diamond distances for $T = \frac{2\eta}{1+\eta}q$ and $G = \frac{1+\eta}{2\eta}$.
\begin{lemma}
\label{lem:trace-channel-loss-channel-diamond-norm-bound}
    Given $K \leq M - s$, 
    \begin{equation}
        \KDiaNorm{\traceChannel - \Loss_T} \leq 2\sqrt{\frac{\K (\K-1)}{2(M - 1)(M - \K + 1)} + \frac{\K}{T(1-T)} \left( \frac{s}{M} - T\right)^2}
        \label{eqn:loss_bound}
    \end{equation}
    for some $ 0 \leq \K \leq K$.
\end{lemma}
\begin{proof}
        Let $V_s, W_T$ be the Stinespring dilations of the loss channel and Gaussian attenuation channel with respect to the Kraus operator ordering specified above:
    \begin{align*}
        V_s = \sum_{u = 0}^{M - s} L_{M \to s}^u \otimes \ket{u},\qquad W_T = \sum_{u = 0}^{\infty} K_T^u \otimes \ket{u} 
    \end{align*}

    Let $\ket{\sigma} = \sum_{\ell = 0}^{K } \sqrt{s_\ell} \ket{\sigma_\ell}\otimes \ket{\sigma_\ell}$ be the pure state achieving the diamond norm $\KDiaNorm{\traceChannel - \Loss_T} = \norm{\traceChannel(\sigma) - \Loss_T(\sigma)}_1$. Then by monotonicity of trace norm,
    \begin{align*}
        \KDiaNorm{\traceChannel - \Loss_T} &\leq \norm{ (\id \otimes V_s)\ketbra{\sigma}{\sigma} (\id \otimes V_s^\dagger) - (\id \otimes W_T)\ketbra{\sigma}{\sigma} (\id \otimes W_T^\dagger) }_1 \\
        &= 2\sqrt{1 - \abs{ \bra{\sigma} (\id \otimes V_s^\dagger W_T) \ket{\sigma} }^2}
    \end{align*}
    where the last line follows from the Fuchs-van de Graaf inequality, which is an equality for pure states. Now let 
    \begin{align*}
        \Jloss &\coloneqq  \bra{\sigma} (\id \otimes V_s^\dagger W_T) \ket{\sigma} = \sum_{\ell = 0}^{K } s_\ell \sum_{u = 0}^{M-s} \bra{\sigma_\ell} L_{M\to s}^{u \, \dagger} K_T^u \ket{\sigma_\ell}
    \end{align*}
Each $ L_{M\to s}^{u \, \dagger} K_T^u$ is diagonal in the Fock basis with non-negative eigenvalues, so for all $\ell$ we have $\bra{\sigma_\ell} L_{M\to s}^{u \, \dagger} K_T^u \ket{\sigma_\ell}\geq 0$, and 
\begin{align}
    \Jloss =  \sum_{\ell = 0}^{K } s_\ell \sum_{u = 0}^{M-s} \bra{\sigma_\ell} L_{M\to s}^{u \, \dagger} K_T^u \ket{\sigma_\ell} &= \sum_{\ell,k = 0}^{K } s_\ell\abs{\langle{\sigma_\ell}|k\rangle}^2  \bra{k} \sum_{u=0}^{M-s} L_{M\to s}^{u \, \dagger} K_T^u \ket{k}\nonumber\\
    &\geq \min_{\ell\in \{0,\dots,K\}} \sum_{k=0}^K\abs{\langle{\sigma_\ell}|k\rangle}^2 \bra{k}\sum_{u=0}^{M-s} L_{M\to s}^{u \, \dagger} K_T^u\ket{k}\nonumber\\
    &\geq \min_{k\in \{0,\dots,K\}}\; \bra{k} \sum_{u=0}^{M-s} L_{M\to s}^{u \, \dagger} K_T^u \ket{k}.
\end{align}
Let us assume that the minimum occurs at $k=\tilde{K}$. In the following, we use this lower bound on $\Jloss$ to prove the claimed upper bound on the diamond distance. As $\tilde{K}\leq K \leq M-s$, we have
\begin{align}
    \Jloss &\geq \sum_{u = 0}^{\K} \sqrt{ \frac{\binom{s}{\K - u} \binom{M - s}{u}}{\binom{M}{\K}} }\sqrt{\binom{\K}{u} T^{\K-u} (1-T)^u } \label{eqn:loss_min}  = \sum_{u = 0}^{\K} \sqrt{\alpha_{\K - u}}\sqrt{\beta_{\K - u}} = \sum_{u = 0}^{\K} \sqrt{\alpha_{u}}\sqrt{\beta_{u}},
\end{align}
 with distributions
    \begin{align*}
        \alpha_u&= \frac{\binom{s}{u} \binom{M - s}{\K - u} }{\binom{M}{\K}} &&\sim \Hypergeo(M, s, \K), \\
        \beta_u &= \binom{\K}{u}T^u (1-T)^{\K - u} &&\sim \Binom(\K, T).
    \end{align*}
Noting that $\text{supp}(\alpha_u) \subset \text{supp}(\beta_u)$, we apply Jensen's inequality to introduce KL divergence as a lower bound~\cite{Tsybakov2009-gw}:
    \begin{align*}
        \sum_{u = 0}^{\K} \sqrt{\alpha_u}\sqrt{\beta_u} = \sum_{u=0}^{\K} \alpha_u \exp\left( -\frac{1}{2} \ln\frac{\alpha_u}{\beta_u}\right)&\geq \exp \left(  -\frac{1}{2}\sum_{u=0}^{\K} \alpha_u  \ln\frac{\alpha_u}{\beta_u} \right) = \exp \left(-\frac{1}{2} \diver{\alpha_u}{\beta_u} \right).
    \end{align*}
 We now turn to upper bounding $\diver{\alpha_u}{\beta_u}$. To see the asymptotic behavior of hypergeometric tending towards binomial, we introduce an intermediary binomial distribution 
    \begin{equation*}
        \gamma_u = \binom{\K}{u} \left( \frac{s}{M} \right)^u \left(1 - \frac{s}{M} \right)^{\K - u} \sim \Binom(\K, \frac{s}{M}).
    \end{equation*}
Then,
    \begin{align}
        \diver{\alpha_u}{\beta_u} &= \diver{\alpha_u}{\gamma_u} + \sum_{u = 0}^{\K} \alpha_u \ln\frac{\gamma_u}{\beta_u}.
        \label{eqn:loss_div}
    \end{align}
Comparison between $\gamma_u$ and $\alpha_u$ is simply the comparison between sampling with and without replacement. This was precisely studied by Stam in~\cite{Stam1978-we}. In Section 2, eq.~(2.6) therein, the following bound is stated:
\begin{align}
    \diver{\alpha_u}{\gamma_u} \leq \frac{\K (\K-1)}{2(M - 1)(M - \K + 1)}.
\end{align}
    As for the second term of Eq.~\eqref{eqn:loss_div}, note that
    \begin{align*}
        \ln\frac{\gamma_u}{\beta_u} &= \ln \frac{\binom{\K}{u} \left(\frac{s}{M} \right)^{u}\left(1 - \frac{s}{M} \right)^{u}}{\binom{\K}{u} T^{\K - u} (1 - T)^{\K - u}} \\
        &= u\ln\left( \frac{\left(\frac{s}{M}\right)}{T}\right) + (\K - u) \ln\left( \frac{1 - \frac{s}{M}}{1 - T} \right).
    \end{align*}
Therefore,
 \begin{align*}
         \sum_{u = 0}^{\K} \alpha_u \ln\frac{\gamma_u}{\beta_u} &= \E_{\alpha_u}[u]\ln\left( \frac{\frac{s}{M}}{T}\right) + (\K - \E_{\alpha_u}[u]) \ln\left( \frac{1 - \frac{s}{M}}{1 - T} \right) \\
         &= \K \frac{s}{M}\ln\left( \frac{\frac{s}{M}}{T}\right) + (\K - \K \frac{s}{M}) \ln\left( \frac{1 - \frac{s}{M}}{1 - T} \right) \\
         &= \K \diver{\frac{s}{M}}{T}\\
         &\leq \frac{\K}{T(1-T)} \left( \frac{s}{M} - T \right)^2.
    \end{align*}
    In the last line we have used a simple upper bound on the binary relative entropy, $\diver{a}{b} \leq \frac{(a-b)^2}{b(1-b)}$, which can be easily seen using $\ln x \leq x-1$. Combining everything, we have a lower bound:
    \begin{align}
        \Jloss &\geq \exp \left({-\frac{1}{4}\frac{\K (\K-1)}{(M - 1)(M - \K + 1)} - \frac{1}{2}\frac{\K}{T(1-T)} \left( \frac{s}{M} - T \right)^2} \right)\nonumber\\
        \implies (\Jloss)^2 &\geq \exp \left({-\frac{1}{2}\frac{\K (\K-1)}{(M - 1)(M - \K + 1)} - \frac{\K}{T(1-T)} \left( \frac{s}{M} - T \right)^2} \right)\nonumber\\
        &\geq 1 -\frac{1}{2}\frac{\K (\K-1)}{(M - 1)(M - \K + 1)} - \frac{\K}{T(1-T)} \left( \frac{s}{M} - T \right)^2,
    \end{align}
where we have used that $e^{-x}\geq 1-x$. This provides the claimed upper bound on the diamond norm.
\end{proof}
We now similarly derive the following upper bound for the cloning channel and the Gaussian amplification channel.
\begin{lemma}
\label{lem:clone-channel-gain-channel-diamond-norm-bound}
    Given $K \leq s$,
    \begin{equation}
        \KDiaNorm{\cloneChannel - \Gain_G} \leq 2\sqrt{\frac{(\K' + 1)(\K'+2)}{(s+2)(s + 1 - \K')} + G(\K' + 1) \abs{\frac{1}{G - 1}\frac{N - s}{s + 2} - 1 } \abs{ \frac{1}{G} - \frac{s + 1}{N + 1} } },
        \label{eqn:gain_bound}
    \end{equation}
    for some $ 0 \leq \K' \leq K $.
\end{lemma}
\begin{proof}
    A similar analysis as in the proof of Lemma~\ref{lem:trace-channel-loss-channel-diamond-norm-bound} leads to 
    \begin{align}
        \KDiaNorm{\cloneChannel - \Gain_G} \leq 2\sqrt{1-J^2_{\mathrm{gain},s}},
    \end{align}  
    where
    \begin{align}
        J_{\mathrm{gain},s} &\geq \sum_{t = 0}^{N - s} \sqrt{\frac{s + 1}{N + 1} \frac{\binom{s}{\K'} \binom{N - s}{t}}{\binom{N}{\K' + t}}} \sqrt{\binom{\K' + t}{t} \frac{(G - 1)^t}{G^{\K' + t + 1}} }\\
        &= \sum_{t =0}^{N - s} \sqrt{\alpha_t'}\sqrt{\beta_t'} \geq \exp\left( -\frac{1}{2}\diver{\alpha_t'}{\beta'_t} \right),
    \end{align}
    where we introduced a negative hypergeometric distribution and a negative binomial distribution
    \begin{align*}
        \alpha_t' &= \frac{\binom{\K' + t}{\K'} \binom{N - \K' - t}{s - \K'} }{\binom{N + 1}{s + 1}} && \sim \NegHypergeo(N + 1, s + 1, \K' + 1) \\
        \beta'_t &= \binom{\K' + t}{t} \left(1 - \frac{1}{G} \right)^t \left(\frac{1}{G} \right)^{\K' + 1} &&\sim \NegBinom(\K' + 1, \frac{1}{G}).
    \end{align*}
    Similarly, as before, we introduce an intermediary negative binomial distribution $$\gamma'_t = \binom{\K' + t}{t}\left(1 - \frac{s + 1}{N + 1} \right)^t \left( \frac{s + 1}{N + 1} \right)^{\K'+ 1}\sim \NegBinom(\K' + 1, \frac{s + 1}{N + 1}),$$
    and decompose the divergence similarly to \eqref{eqn:loss_div}:
    \begin{equation}
        \diver{\alpha_t'}{\beta_t'} = \diver{\alpha_t'}{\gamma_t'} + \sum_{t = 0}^{N - s} \alpha_t' \ln \frac{\gamma_t'}{\beta_t'}
        \label{eqn:gain_div}
    \end{equation}
    The first term can be bounded again via Stam's inequality (\cite[eq.~(4.1)]{Stam1978-we}) as
    \begin{equation}
        \diver{\alpha_t'}{\gamma_t'} \leq \frac{(\K' + 1)(\K'+2)}{2(s+2)(s + 1 - \K')}.
        \label{eqn:gain_first}
    \end{equation}
    The second term in \eqref{eqn:gain_div} evaluates to
    \begin{align*}
        \sum_{t = 0}^{N - s} \alpha_t' \ln \frac{\gamma_t'}{\beta_t'} &= (\K' + 1 )\ln\left( \frac{ \frac{s + 1}{N + 1} }{ \frac{1}{G}} \right) + \E_{\alpha_t'}[t] \ln\left( \frac{1 - \frac{s + 1}{N + 1}}{ 1 - \frac{1}{G} } \right) \\
        &=(\K' + 1) \ln \left(G \frac{s + 1}{N + 1} \right) + \frac{(\K' + 1)(N - s)}{s + 2} \ln \left( \frac{1 - \frac{s + 1}{N + 1}}{1 - \frac{1}{G}} \right) \\
        &= (\K' + 1) \left( \ln\left( G \frac{s + 1}{N + 1} \right) + \frac{N - s}{s + 2} \ln \left( \frac{1 - \frac{s + 1}{N + 1}}{1 - \frac{1}{G}} \right) \right) \\
        & \leq (\K' + 1) \left( \left( G \frac{s + 1}{N + 1} - 1 \right) + \frac{N - s}{s + 2} \left( \frac{1 - \frac{s + 1}{N + 1}}{1 - \frac{1}{G}} - 1\right) \right) \\
        &= (\K' + 1) \left( G \left(  \frac{s + 1}{N + 1} - \frac{1}{G} \right) + \frac{1}{1 - \frac{1}{G}}\frac{N - s}{s + 2} \left(1 - \frac{s + 1}{N + 1} - (1 - \frac{1}{G})\right) \right) \\
        &= G(\K' + 1) \left( \left(  \frac{s + 1}{N + 1} - \frac{1}{G} \right) + \frac{1}{G - 1}\frac{N - s}{s + 2} \left( \frac{1}{G} - \frac{s + 1}{N + 1} \right) \right) \\
        &= G(\K' + 1) \left(\frac{1}{G - 1}\frac{N - s}{s + 2} - 1 \right) \left( \frac{1}{G} - \frac{s + 1}{N + 1} \right) \\
        &\leq G(\K' + 1) \abs{\frac{1}{G - 1}\frac{N - s}{s + 2} - 1 } \abs{ \frac{1}{G} - \frac{s + 1}{N + 1} },
    \end{align*}
    where we have used $\ln x \leq x - 1$ to obtain the inequality. 
    Combining everything and finally using $e^{-x}\geq 1-x$ gives the stated upper bound on the diamond distance.
\end{proof}
We are now ready to give the proof of the main result in this section:
\begin{proof}[Proof of Theorem~\ref{thm:diamond-dist-bound}]
    Using Eq~\eqref{eq:tri-diamond} and Eq~\eqref{eq:decomp-into-loss-and-gain}, we obtain that 
    \begin{align}
        &\;\norm{\PhiChannel - \mathcal{A}_{\frac{1+\eta}{2\eta}}\circ \mathcal{L}_{\frac{2\eta}{1+\eta}\frac{N}{M}}}_{\diamond, K}\\
    &\leq \sum_{s=0}^N \piDist{s} \norm{\cloneChannel
    \circ
    \traceChannel- \mathcal{A}_{\frac{1+\eta}{2\eta}}\circ \mathcal{L}_{\frac{2\eta}{1+\eta}\frac{N}{M}}}_{\diamond, K}\nonumber\\
    &\leq  \sum_{s=0}^N \piDist{s} \left(\KDiaNorm{\cloneChannel - \Gain_{\frac{1+\eta}{2\eta}}} +  \KDiaNorm{\traceChannel - \Loss_{\frac{2\eta}{1+\eta}\frac{N}{M}}}\right)\nonumber\\
    &\leq \smashoperator{\sum_{s\in \mathrm{Typ}^N(\delta)} }\piDist{s}\left(\KDiaNorm{\cloneChannel - \Gain_{\frac{1+\eta}{2\eta}}} +  \KDiaNorm{\traceChannel - \Loss_{\frac{2\eta}{1+\eta}\frac{N}{M}}}\right) + 4 \smashoperator{\sum_{s\notin \mathrm{Typ}^N(\delta)}} \piDist{s}.
    \end{align}
    Here $\mathrm{Typ}^N(\delta)$ denotes the interval $\left[\left(\frac{2\eta}{1+\eta} -\delta\right)N, \left(\frac{2\eta}{1+\eta} +\delta\right)N \right]$. We have used the trivial upper bound of $2$ for the diamond distance in the atypical interval for $s$. With the choice $\delta=\sqrt{\frac{\log N}{N}}$, Lemma~\ref{lem:mixing-weight-concentration} applies for all $N\geq 2$. Hence, we obtain that
    \begin{multline}
         \norm{\PhiChannel - \mathcal{A}_{\frac{1+\eta}{2\eta}}\circ \mathcal{L}_{\frac{2\eta}{1+\eta}\frac{N}{M}}}_{\diamond, K}\\
        \leq \sum_{s\in \mathrm{Typ}^N(\delta)} \piDist{s} \left(\KDiaNorm{\cloneChannel - \Gain_{\frac{1+\eta}{2\eta}}} +  \KDiaNorm{\traceChannel - \Loss_{\frac{2\eta}{1+\eta}\frac{N}{M}}}\right) + 4c\frac{1+\eta}{\eta}\;e^{-\frac{2N^2 \delta^2}{N+1}}.\label{eq:total-bound-1}
    \end{multline}
Within the typical interval, it is sufficient to give a uniform upper bound for each of the diamond distance term. We begin with the cloning and amplification channel and use the bound \eqref{eqn:gain_bound} from Lemma \ref{lem:clone-channel-gain-channel-diamond-norm-bound}. For $s\in \mathrm{Typ}^N(\delta)$ with the choice $\delta = \sqrt{\frac{\log N}{N}}$ and $N>\max\{N_0,\frac{1+\eta}{\eta}K\}$ with \begin{align}
    N_0 &= \frac{8}{3}\left(\frac{1+\eta}{\eta}\right)^2\log \frac{2+2\eta}{\eta},
\end{align} we have that $\frac{2\eta}{1+\eta} -\sqrt{\frac{\log N}{N}}>\frac{\eta}{1+\eta}$ (using $N>N_0$) and $s\geq N(\frac{2\eta}{1+\eta}-\sqrt{\frac{\log N}{N}}) > \frac{1+\eta}{\eta} K \frac{\eta}{1+\eta} =K \geq \K'$ (using $N> \frac{1+\eta}{\eta}K$).
Therefore,
\begin{align}
    \frac{1}{(s+2)(s+1-\tilde{K}')}& \leq \frac{1}{s} \leq \frac{1}{N(\frac{2\eta}{1+\eta}-\sqrt{\frac{\log N}{N}})} \leq \frac{1+\eta}{qM\eta}.
\end{align}
Furthermore,
\begin{align}
    &\frac{N - s}{s + 2} \leq \frac{N-s}{s} = \frac{1-s/N}{s/N} \leq \frac{1-(\frac{2\eta}{1+\eta}-\sqrt{\frac{\log N}{N}})}{\frac{2\eta}{1+\eta}-\sqrt{\frac{\log N}{N}}} \leq \frac{1-\frac{\eta}{1+\eta}}{\frac{\eta}{1+\eta}} = \frac{1}{\eta}\nonumber\\
    \implies & \abs{\frac{1}{\frac{1+\eta}{2\eta} - 1}\frac{N - s}{s + 2} - 1 } \leq \frac{2\eta}{1-\eta}\frac{1}{\eta} + 1 \leq \frac{3-\eta}{1-\eta},
\end{align}
and
\begin{align}
    \abs{ \frac{s+1}{N+1} - \frac{2\eta}{1+\eta} }  &\leq \abs{\frac{\frac{2\eta}{1+\eta}N+\delta N+1}{N+1} -\frac{2\eta}{1+\eta}}\nonumber\\
    &\leq \abs{\delta \frac{N}{N+1} + \left(\frac{1-\eta}{1+\eta}\right)\frac{1}{N+1}}\nonumber\\
    &\leq \left(\sqrt{\frac{\log qM}{qM}} + \frac{1-\eta}{1+\eta}\frac{1}{qM}\right).
\end{align}
Combining both terms, we obtain the following upper bound from Eq.~\eqref{eqn:gain_bound}:
For $s\in \mathrm{Typ}^N(\delta)$, we have $\big\|\cloneChannel - \Gain_{\frac{1+\eta}{2\eta}}\big\|_{\diamond,K} \leq \varepsilon^{\mathrm{cloning}}_{M,K,q}$ with
\begin{align}
    \varepsilon^{\mathrm{cloning}}_{M,K,q}= 2 \left(\frac{1}{M}\frac{(K+2)^2(1+\eta)}{q\eta} + \frac{(1+\eta)(3-\eta)(K+1)}{2\eta(1-\eta)}\left(\sqrt{\frac{\log qM}{qM}} + \frac{1-\eta}{1+\eta}\frac{1}{qM}\right) \right)^{1/2}.
\end{align}
Evidently, we have $\varepsilon^{\mathrm{cloning}}_{M,K,q}\to 0$ as $M\to \infty$ for fixed $K$ and $q$. 

The bound on the diamond norm between loss channel and the Gaussian attenuation channel is even simpler. In Eq.~\eqref{eqn:loss_bound} in Lemma \ref{lem:trace-channel-loss-channel-diamond-norm-bound}, we write $T = \frac{2\eta}{1+\eta} q$ and $s/M =sq/N$ and within the typical interval for $s$, we have that $(\frac{s}{M} - T)^2 = q^2(\frac{s}{N}-\frac{2\eta}{1+\eta})^2 \leq q^2\frac{\log N}{N} = q\frac{\log qM}{M}$. Therefore, the second term within the square-root is upper bounded by $$\frac{(1+\eta)K}{2\eta q(1-\frac{2\eta}{1+\eta}q)}q\frac{\log qM}{M} = \frac{(1+\eta)K}{2\eta (1-\frac{2\eta}{1+\eta}q)}\frac{\log qM}{M}.$$
Hence, a simple upper bound on the loss and the Gaussian attenuation channel within the typical window of $s$ is $\big\|\traceChannel - \Loss_{\frac{2\eta}{1+\eta}\frac{N}{M}} \big\|_{\diamond,K}\leq \varepsilon^{\mathrm{loss}}_{M,K,q}$ with
\begin{align}
\varepsilon^{\mathrm{loss}}_{M,K,q} = 2\left( \frac{K^2}{2(M-1)(M-K+1)} + \frac{(1+\eta)K}{2\eta (1-\frac{2\eta}{1+\eta}q)}\frac{\log qM}{M}\right)^{1/2}.
\end{align}
From the expression above it is clear that for fixed $K$ and $q$ as given in the statement of the theorem, $\varepsilon^{\mathrm{loss}}_{M,K,q} \to 0$ as $M\to \infty$. As $\varepsilon^{\mathrm{cloning}}_{M,K,q}$ and $\varepsilon^{\mathrm{loss}}_{M,K,q}$ bound each of the terms in the respective sums within the typical interval of $s$ in Eq.~\eqref{eq:total-bound-1}, and $\sum_{s\in \mathrm{Typ}^N(\delta)}\piDist{s}\leq 1$, we conclude that
\begin{align}
     \norm{\PhiChannel - \mathcal{A}_{\frac{1+\eta}{2\eta}}\circ \mathcal{L}_{\frac{2\eta}{1+\eta}\frac{N}{M}}}_{\diamond, K} &\leq \varepsilon^{\mathrm{cloning}}_{M,K,q} + \varepsilon^{\mathrm{loss}}_{M,K,q}+ 4c\frac{1+\eta}{\eta}\;e^{-\frac{2N^2 \delta^2}{N+1}}\nonumber\\
    &\leq \varepsilon^{\mathrm{cloning}}_{M,K,q} + \varepsilon^{\mathrm{loss}}_{M,K,q}+ 4c\frac{1+\eta}{\eta}\;e^{-{N \delta^2}}\nonumber\\
    &\leq \varepsilon^{\mathrm{cloning}}_{M,K,q} + \varepsilon^{\mathrm{loss}}_{M,K,q}+ 4c\frac{1+\eta}{\eta q}\frac{1}{M} \eqqcolon \varepsilon_{M,K,q}
\end{align}
It is explicit that for fixed cutoff $K$ and ratio $q=N/M$ we have $\varepsilon_{M,K,q}\to 0$ as $M\to \infty$. This finishes the proof. 
\end{proof}

We also give the proof of Theorem~\ref{thm:diamond-dist-bound-main} in the main text, which is stated in slightly different terms than in Theorem \ref{thm:diamond-dist-bound} above.

\begin{proof}[Proof of Theorem~\ref{thm:diamond-dist-bound-main} in main text]
Here we focus our attention specifically to the typical interval of $q=N/M \in \left[\eta -\delta, \eta +\delta\right]\eqqcolon\mathrm{Typ}^M(\delta)$ with $\delta = \sqrt{(2\log M)/M}$. When $1>\eta \geq 1/2$, for all $M > \frac{16}{(1-\eta)^2}\log \frac{16}{(1-\eta)^2}$, the typical interval, $\mathrm{Typ}^M(\sqrt{(2\log M)/M}) \subset \left[\frac{3\eta-1}{2},\frac{1+\eta}{2}\right]\subset (0,1)$, implying $\frac{3\eta-1}{2}<q< \frac{1+\eta}{2}$. Therefore, for all $q \in \mathrm{Typ}^M(\sqrt{(2\log M)/M})$ we can upper-bound the constant $\varepsilon_{M,K,q}$ by
\begin{align}\label{eq:bound-eta-geq1/2}
    \varepsilon_{M,K}&= 2 \left(\frac{2}{M}\frac{(K+2)^2(1+\eta)}{(3\eta-1)\eta} + \frac{(1+\eta)(3-\eta)(K+1)}{2\eta(1-\eta)}\left(\sqrt{\tfrac{2\log [(1+\eta)M/2]}{(3\eta-1)M}} + \tfrac{1-\eta}{1+\eta}\tfrac{2}{(3\eta-1)M}\right) \right)^{1/2}\nonumber\\
    &\quad {}+2\left( \frac{K^2}{2(M-1)(M-K+1)} + \frac{(1+\eta)K}{2\eta (1-\eta)}\frac{\log [(1+\eta)M/2]}{M}\right)^{1/2} + \frac{8c(1+\eta)}{\eta (3\eta-1)}\frac{1}{M}.
\end{align}
When $0<\eta \leq 1/2$, for all $M > \frac{16}{\eta^2}\log \frac{16}{\eta^2}$, the typical interval, $\mathrm{Typ}^M(\sqrt{(2\log M)/M}) \subset \left[\frac{\eta}{2},\frac{3\eta}{2}\right]\subset (0,1)$, implying $\frac{\eta}{2}<q< \frac{3\eta}{2}$. In this case, we can therefore upper-bound $\varepsilon_{M,K,q}$ for all $q \in \mathrm{Typ}^M(\sqrt{(2\log M)/M})$ by
\begin{align}\label{eq:bound-eta-leq1/2}
    \varepsilon_{M,K} &= 2 \left(\frac{2}{M}\frac{(K+2)^2(1+\eta)}{\eta^2} + \frac{(1+\eta)(3-\eta)(K+1)}{2\eta(1-\eta)}\left(\sqrt{\tfrac{2\log (3\eta M/2)}{\eta M}} + \tfrac{1-\eta}{1+\eta}\tfrac{2}{\eta M}\right) \right)^{1/2}\nonumber\\
    &\quad {} +2\left( \frac{K^2}{2(M-1)(M-K+1)} + \frac{(1+\eta)^2K}{2\eta (1+\eta-3\eta^2)}\frac{\log (3\eta M/2)}{M}\right)^{1/2} + \frac{8c(1+\eta)}{\eta^2}\frac{1}{M}.
\end{align}
For both of the cases in Eq.~\eqref{eq:bound-eta-geq1/2} and Eq.~\eqref{eq:bound-eta-leq1/2}, we find that $\varepsilon_{M,K}\to 0$ as $M\to \infty$ for a fixed cutoff $K$.
\end{proof}

\subsection{Convergence among the typical Gaussian channels}
\begin{lemma}
\label{lem:typical-gaussian-channels-convergence}
    For any bipartite pure state $\ket{\psi}_{RA}$ with $\rho = \tr_R(\ketbra{\psi}{\psi})$ supported within the first $K$ excitations and $N/M \in [\eta - \delta ,\eta + \delta]$ with $\delta = \sqrt{(2\log M)/M}$, we have
    \begin{align}
        \norm{\Id_R\otimes \mathcal{A}_{\frac{1+\eta}{2\eta}}\circ \mathcal{L}_{\frac{2\eta}{1+\eta}\frac{N}{M}}(\psi) - \Id_R\otimes \mathcal{A}_{\frac{1+\eta}{2\eta}}\circ \mathcal{L}_{\frac{2\eta^2}{1+\eta}}(\psi)}_1 \leq \varepsilon^{\mathrm{att}}_{M,K},
    \end{align}
    where $\varepsilon^{\mathrm{att}}_{M,K}\to 0$ as $M\to \infty$. 
\end{lemma}
\begin{proof}
Consider the following:
\begin{align}
    &\norm{\Id_R\otimes \mathcal{A}_{\frac{1+\eta}{2\eta}}\circ \mathcal{L}_{\frac{2\eta}{1+\eta}\frac{N}{M}}(\psi) - \Id_R\otimes \mathcal{A}_{\frac{1+\eta}{2\eta}}\circ \mathcal{L}_{\frac{2\eta^2}{1+\eta}}(\psi)}_1 \nonumber\\
    &\leq  \norm{\Id_R\otimes\mathcal{L}_{\frac{N}{M}}(\psi) - \Id_R\otimes\mathcal{L}_{\eta}(\psi)}_1\nonumber\\
    &\leq \norm{\left( \one_R\otimes V_{\frac{N}{M}}\right)\ketbra{\psi}{\psi}\left( \one_R\otimes V^\dagger_{\frac{N}{M}} \right)- \left( \one_R\otimes V_\eta \right)\ketbra{\psi}{\psi} \left( \one_R\otimes V^\dagger_{\eta}\right) }_1\nonumber\\
    &= 2\sqrt{1- |\bra{\psi} \left(\one_R\otimes V^\dagger_\eta V_{\frac{N}{M}}\right)\ket{\psi}|^2},
\end{align}
where we used $V_{T} =  \sum_{u} K^u_T\otimes \ket{u}$ with $K_T^u$ as in \eqref{eq:loss-channel-kraus-operators} in the second line.
The first and second inequality follows from monotonicity of trace distance and the composition rule $\mathcal{L}_a\circ\mathcal{L}_b=\mathcal{L}_{ab}$, while the last equality is the equality case for the Fuchs-van de Graaf inequality for pure states. Let $\ket{\psi} = \sum_{\alpha}\sqrt{\lambda_\alpha}\ket{\alpha}\ket{\alpha}$ be the Schmidt decomposition of $|\psi\rangle$. Then,
\begin{align}
    \bra{\psi} \left(\one_R\otimes V^\dagger_\eta V_{\frac{N}{M}}\right)\ket{\psi} = \sum_{\alpha} \lambda_\alpha \bra{\alpha} V^\dagger_\eta V_{\frac{N}{M}} \ket{\alpha} = \sum_{\alpha} \lambda_\alpha \sum_{u} \bra{\alpha} K^{u\dagger}_\eta K^u_{\frac{N}{M}}\ket{\alpha} 
\end{align}
Writing $\ket{\alpha} = \sum_{k}f_{\alpha,k} \ket{k}$ and using definition \eqref{eq:loss-channel-kraus-operators} for the $K^u_*$, we continue to calculate:
\begin{align}
    &\bra{\psi} \left(\one_R\otimes V^\dagger_\eta V_{\frac{N}{M}}\right)\ket{\psi}\nonumber\\
    &= \sum_{\alpha}\lambda_\alpha \sum_{k} \abs{f_{\alpha,k}}^2 \sum_{u=0}^k\sqrt{{k\choose u}\eta^{(k-u)}(1-\eta)^{u}}\sqrt{{k\choose u}\left(\frac{N}{M}\right)^{(k-u)}\left(1-\frac{N}{M}\right)^{u}}\\
    &\geq \min_{k\in \{0,\dots , K\}} \sum_{u=0}^k{k\choose u}\sqrt{\eta^{(k-u)}(1-\eta)^{u}}\sqrt{\left(\frac{N}{M}\right)^{(k-u)}\left(1-\frac{N}{M}\right)^{u}}.\\
    &\eqqcolon J.
\end{align}
Now, for every $k$ the binomial theorem gives
\begin{align}
    \sum_{u=0}^k{k\choose u}\sqrt{\eta^{(k-u)}(1-\eta)^{u}}\sqrt{\left(\frac{N}{M}\right)^{(k-u)}\left(1-\frac{N}{M}\right)^{u}} 
    =\left(\sqrt{\frac{N\eta}{M}}+\sqrt{\left(1-\frac{N}{M}\right)(1-\eta)}\right)^k . 
\end{align}
By the Cauchy-Schwarz inequality,
\begin{align} 
    \sqrt{\frac{N\eta}{M}}+\sqrt{\left(1-\frac{N}{M}\right)(1-\eta)} \leq \left(\frac{N}{M}+1-\frac{N}{M}\right)(\eta + 1-\eta) = 1,
\end{align}
and hence the minimum occurs at the cutoff $k=K$.
Therefore,
\begin{align}
    J &=  \left(\sqrt{\frac{N\eta}{M}}+\sqrt{\left(1-\frac{N}{M}\right)(1-\eta)}\right)^K\nonumber\\
    & \geq \left(\eta\sqrt{1-\frac{\delta}{\eta}} +(1-\eta)\sqrt{\left(1-\frac{\delta}{1-\eta}\right)} \right)^K \nonumber\\
    &\geq \left(\eta\left(1-\frac{\delta}{\eta}\right) + (1-\eta)\left(1-\frac{\delta}{1-\eta}\right)\right)^K\nonumber\\
    &= \left(1 -2\delta\right)^K
\end{align}
for sufficiently large $M$ such that $\delta=\sqrt{(2\log M)/M}< \min\lbrace \eta,1-\eta,1/2\rbrace$. Therefore, using Bernoulli's inequality, we have $J \geq 1- 2K\delta$, and
\begin{align}
    \norm{\Id_R\otimes \mathcal{A}_{\frac{1+\eta}{2\eta}}\circ \mathcal{L}_{\frac{2\eta}{1+\eta}\frac{N}{M}}(\psi) - \Id_R\otimes \mathcal{A}_{\frac{1+\eta}{2\eta}}\circ \mathcal{L}_{\frac{2\eta^2}{1+\eta}}(\psi)}_1 
    & \leq 2\sqrt{1-(1-2K\delta)^2}\nonumber\\
    & \leq 2 \sqrt{4K\delta - 4K^2\delta^2}\nonumber\\
    &\leq 4\sqrt{K\delta}\sqrt{1-K\delta}\nonumber\\
    &\leq 4 \sqrt{K}\left(\frac{2\log M}{M\nonumber}\right)^{1/4}\\
    &\eqqcolon \varepsilon^{\mathrm{att}}_{M,K}, 
    \label{eq:eps-att-MK}
\end{align}
for sufficiently large $M$ such that $K\delta \leq \frac{1}{2}$. We see that $\varepsilon^{\mathrm{att}}_{M,K}\to 0$ as $M\to \infty$.
\end{proof}

\subsection{Convergence of coherent information}
In this section we prove Theorem~\ref{thm:coherent-information} from the main text, which we restate here for convenience:
\begin{theorem}\label{thm:coherent-information-supp} For every $\rho\in \mathcal{L}(\mathrm{Sym}^M(\mathbb{C}^2))$ supported within a fixed Dicke excitation cutoff (equivalently Fock cutoff) given by $K$, 
    \begin{align}
        \abs{I_c(\rho,\mathcal{D}^{\otimes M}_\eta) - I_c\left(\rho, \mathcal{A}_{\frac{1+\eta}{2\eta}}\circ \mathcal{L}_{\frac{2\eta^2}{1+\eta}}\right)} \leq \Delta_{M,K},
    \end{align}
    where $\Delta_{M,K}\to 0$ as $M\to \infty$.
\end{theorem}
\begin{proof} 
Using the decomposition in Theorem~\ref{thm:decomp}, we can write $I_c(\rho,\mathcal{D}^{\otimes M}_\eta) = \sum_N \pDist{N} I_c(\rho, \PhiChannel )$. Then adding and subtracting $\sum_N \pDist{N} I_c(\rho, \mathcal{A}_{\frac{1+\eta}{2\eta}}\circ \mathcal{L}_{\frac{2\eta}{1+\eta}\frac{N}{M}})$ and using the triangular inequality, we obtain
    \begin{align}
        &\abs{I_c(\rho,\mathcal{D}^{\otimes M}_\eta) - I_c\left(\rho, \mathcal{A}_{\frac{1+\eta}{2\eta}}\circ \mathcal{L}_{\frac{2\eta^2}{1+\eta}}\right)}\nonumber\\
        &\leq  \sum_{N} \pDist{N} \abs{I_c(\rho, \PhiChannel) - I_c\left(\rho, \mathcal{A}_{\frac{1+\eta}{2\eta}}\circ \mathcal{L}_{\frac{2\eta}{1+\eta}\frac{N}{M}}\right)}\nonumber \\
        &\quad {} + \sum_{N} \pDist{N} \abs{I_c\left(\rho, \mathcal{A}_{\frac{1+\eta}{2\eta}}\circ \mathcal{L}_{\frac{2\eta}{1+\eta}\frac{N}{M}}\right) - I_c\left(\rho, \mathcal{A}_{\frac{1+\eta}{2\eta}}\circ \mathcal{L}_{\frac{2\eta^2}{1+\eta}}\right)}. \nonumber
    \end{align}
    We divide each of the sums into the typical sector, $N\in \bigl[\eta M - \sqrt{2M\log M}, \eta M + \sqrt{2M \log M} \bigr]$ and the atypical sector, i.e., the complement of this interval. For the atypical sector, we upper-bound the absolute difference between the coherent information by $2\log(K+1)$ and use the concentration inequality from Lemma~\ref{lem:block-probability-concentration}:
    \begin{align}
        &\abs{I_c(\rho,\mathcal{D}^{\otimes M}_\eta) - I_c\left(\rho, \mathcal{A}_{\frac{1+\eta}{2\eta}}\circ \mathcal{L}_{\frac{2\eta^2}{1+\eta}}\right)}\nonumber\\
        & \leq \sum_{N\in \mathrm{Typ}_M} \pDist{N} \abs{I_c(\rho, \Phi^\eta_{M\to N}) - I_c\left(\rho, \mathcal{A}_{\frac{1+\eta}{2\eta}}\circ \mathcal{L}_{\frac{2\eta}{1+\eta}\frac{N}{M}}\right)} \nonumber\\
        & \quad {}+  \sum_{N\in \mathrm{Typ}_M} \pDist{N} \abs{I_c\left(\rho, \mathcal{A}_{\frac{1+\eta}{2\eta}}\circ \mathcal{L}_{\frac{2\eta}{1+\eta}\frac{N}{M}}\right) - I_c\left(\rho, \mathcal{A}_{\frac{1+\eta}{2\eta}}\circ \mathcal{L}_{\frac{2\eta^2}{1+\eta}}\right)} \nonumber\\
        &\quad {} + \frac{4(1+\eta)\log(K+1)}{\eta M}.
    \end{align}
For sums involving the typical sectors, we use the continuity of conditional entropy~\cite{Winter_2016,Alicki_2004}. Then for the first sum, we use the bound on the diamond norm distance in Theorem~\ref{thm:diamond-dist-bound-main} to conclude
\begin{multline}
    \sum_{N\in \mathrm{Typ}_M} \pDist{N} \abs{I_c(\rho, \PhiChannel ) - I_c\left(\rho, \mathcal{A}_{\frac{1+\eta}{2\eta}}\circ \mathcal{L}_{\frac{2\eta}{1+\eta}\frac{N}{M}}\right)} \\
    \leq  2\varepsilon_{M,K}\log(K+1) + (1+\varepsilon_{M,K})\;h_2\left(\frac{\varepsilon_{M,K}}{1+\varepsilon_{M,K}}\right).  
\end{multline}
For the second sum,  we similarly bound the coherent information in terms of $\varepsilon^{\mathrm{att}}_{M,K}$ defined in \eqref{eq:eps-att-MK} in the proof of Lemma \ref{lem:typical-gaussian-channels-convergence}, which gives a bound on the maximum diamond distance (with fixed cutoff $K$) between the attenuation channels $\mathcal{L}_{N/M}$ and $\mathcal{L}_{\eta}$ within the typical sector. Hence,
\begin{multline}
    \sum_{N\in \mathrm{Typ}_M} \pDist{N} \abs{I_c\left(\rho, \mathcal{A}_{\frac{1+\eta}{2\eta}}\circ \mathcal{L}_{\frac{2\eta}{1+\eta}\frac{N}{M}}\right) - I_c\left(\rho, \mathcal{A}_{\frac{1+\eta}{2\eta}}\circ \mathcal{L}_{\frac{2\eta^2}{1+\eta}}\right)} \\
    \leq 2\varepsilon^{\mathrm{att}}_{M,K}\log(K+1) + (1+\varepsilon^{\mathrm{att}}_{M,K})\;h_2\left(\frac{\varepsilon^{\mathrm{att}}_{M,K}}{1+\varepsilon^{\mathrm{att}}_{M,K}}\right). 
\end{multline}
Combining everything we obtain an upper bound on the difference of coherent informations as
\begin{align}
    \abs{I_c(\rho,\mathcal{D}^{\otimes M}_\eta) - I_c\left(\rho, \mathcal{A}_{\frac{1+\eta}{2\eta}}\circ \mathcal{L}_{\frac{2\eta^2}{1+\eta}}\right)} &\leq  2\varepsilon_{M,K}\log(K+1) + (1+\varepsilon_{M,K})\;h_2\left(\frac{\varepsilon_{M,K}}{1+\varepsilon_{M,K}}\right) \nonumber\\
    &\quad {}+ 2\varepsilon^{\mathrm{att}}_{M,K}\log(K+1) + (1+\varepsilon^{\mathrm{att}}_{M,K})\;h_2\left(\frac{\varepsilon^{\mathrm{att}}_{M,K}}{1+\varepsilon^{\mathrm{att}}_{M,K}}\right)\nonumber\\
    &\quad {} + \frac{4(1+\eta)\log(K+1)}{\eta M}\\
    &\eqqcolon \Delta_{M,K}
\end{align}
Note that both $\varepsilon_{M,K}$ and $\varepsilon^{\mathrm{att}}_{M,K}$ go to $0$ as $M\to \infty$, and hence we also have $\lim_{M\to \infty}\Delta_{M,K}=0$.
This concludes the proof.
\end{proof}

\clearpage
\endgroup

\end{document}